\documentclass[journal]{IEEEtran}

\ifCLASSINFOpdf
\else
\fi

\usepackage{amsmath,amsfonts}
\usepackage{algorithm}
\usepackage{algpseudocode}
\usepackage{array}
\usepackage[caption=false,font=footnotesize,labelfont=sf,textfont=sf]{subfig}
\usepackage{textcomp}
\usepackage{stfloats}
\usepackage{url}
\usepackage{verbatim}
\usepackage{graphicx}
\usepackage{cite}
\usepackage{amssymb}
\usepackage{amsthm}

\newtheorem{lemma}{Lemma}
\newtheorem{theorem}{Theorem}
\newtheorem{corollary}{Corollary}

\begin{document}
%
\title{EQE-QAOA: An Equivalence-Preserving Qubit Efficient Framework for Combinatorial Optimization\\[0.6em]}

\author{Xiaoyu Ma,
        Fang Fang,~\IEEEmembership{Senior Member,~IEEE,} Ximing Xie,~\IEEEmembership{Member,~IEEE,} Xianbin Wang,~\IEEEmembership{Fellow,~IEEE,} and~Lajos Hanzo,~\IEEEmembership{Life Fellow,~IEEE}
\thanks{Xiaoyu Ma, Ximing Xie and Xianbin Wang are with the Department of Electrical and Computer Engineering, Western University, London, ON N6A 5B9, Canada (e-mail: xma535@uwo.ca; xxie269@uwo.ca; xianbin.wang@uwo.ca).}
\thanks{Fang Fang is with the Department of Electrical and Computer Engineering and the Department of Computer Science, Western University, London, ON N6A 5B9, Canada (e-mail: fang.fang@uwo.ca).}
\thanks{Lajos Hanzo is with the Department of Electronics and Computer Science, University of Southampton, Southampton SO17 1BJ, U.K. (e-mail:lh@ecs.soton.ac.uk).}}

%
%
%

\maketitle

\begin{abstract}
    The limited number of qubits is a major bottleneck in Quantum Approximate Optimization Algorithm (QAOA) for large-scale combinatorial optimization in the Noisy Intermediate-Scale Quantum (NISQ) era. To make progress, existing techniques rely on qubit reduction at the cost of information loss, hence leading to degraded computational performance. As a remedy, we propose the Equivalence-preserving Qubit Efficient QAOA (EQE-QAOA), which significantly reduces the required number of qubits without degrading the performance of QAOA. By exploiting intrinsic symmetries and conserved quantities, we first demonstrate that the QAOA dynamics are strictly confined to an invariant subspace of the Hilbert space. We subsequently prove that the evolution within this subspace is exactly equivalent to that of the full-scale system, achieving the same optimal solution as the original QAOA. Moreover, to reduce the number of qubits, we propose an isometric mapping that re-encodes the subspace into a space relying on fewer qubits. Furthermore, we derive the applicability conditions of EQE-QAOA and show that it is broadly applicable to large-scale combinatorial optimization problems, excluding only unconstrained problems with completely independent variables. Numerical simulations based on Max-Cut instances validate that EQE-QAOA significantly reduces qubit requirements and computational resources, while preserving exact optimization performance.
\end{abstract}


\IEEEpeerreviewmaketitle

\section{Introduction}

\IEEEPARstart{Q}{uantum} Computing (QC)~\cite{steane1998quantum, bhat2022quantum, sood2023quantum, khang2024introduction} has emerged as a competitive candidate for solving combinatorial optimization problems~\cite{cui2022quantum, chicano2025combinatorial, mitsiou2025quantum, krikidis20251, 9963673, 7482755, 11442662} by exploring exponentially large solution spaces. This unique capability arises from exploiting fundamental quantum properties such as superposition and interference, which allow quantum systems to represent and manipulate many candidate solutions simultaneously~\cite{shafique2024quantum, lu2023quantum}. Consequently, QC provides compelling global search capabilities, potentially capable of outperforming classical heuristics in finding high-quality solutions.

Despite this potential, applying quantum theory to practical problems has been a challenge. Current quantum devices operate in the Noisy Intermediate-Scale Quantum (NISQ) era~\cite{preskill2018quantum}, where hardware performance remains constrained by environmental disturbances and decoherence. Under these limitations, the Quantum Approximate Optimization Algorithm (QAOA)~\cite{farhi2014quantum, blekos2024review, cui2022quantum, choi2019tutorial, zhou2023qaoa} has emerged as a critical approach due to its hybrid classical–quantum structure, enabling efficient computation on quantum devices~\cite{guerreschi2019qaoa}. However, the scarcity of usable qubits remains a fundamental bottleneck in near-term quantum computing, despite the rapid progress in hardware development~\cite{preskill2018quantum, kecceci2025accuracy}. Even though modern superconducting and ion-trap devices already provide thousands of physical qubits, two salient reasons still limit their effective use, leaving qubit resources scarce. Firstly, due to quantum hardware noise, limited coherence time (the duration over which quantum states maintain superposition), and restricted qubit connectivity, only around 10–20 qubits can typically be reliably used for practical quantum computations~\cite{brown2021materials, preskill2018quantum}. Secondly, computer simulation of quantum systems is severely constrained by the exponential growth of the Hilbert space, with practical limits around 30 qubits~\cite{young2023simulating, scherer2017computational}. This limitation is further exacerbated in conventional QAOA implementations, where variable-level encoding leads to a number of required qubits that scales linearly with the problem size. More importantly, a recent study proved that achieving a clear quantum advantage with QAOA may require at least hundreds of high-quality qubits~\cite{guerreschi2019qaoa}. Consequently, developing qubit-efficient quantum optimization methods has become a critical research area.

To address this issue, researchers have explored various strategies to reduce qubit requirements. These efforts generally fall into three categories: (i) compact encoding, which represents multiple variables using fewer qubits~\cite{tan2021qubit, podobrii2025qubit, kanatbekova2025qubit}, (ii) problem partitioning, which decomposes a large problem into smaller subproblems with fewer qubits~\cite{zhou2023qaoa, li2022large, hong2025qubit}, (iii) circuit-level compression, which improves qubit utilization to reduce qubit requirements~\cite{gokhale2019asymptotic, decross2023qubit}. Specifically, in the first category, compact encoding strategies redesign encoding schemes to reduce the number of qubits. Tan \textit{et al.}~\cite{tan2021qubit} introduced a Minimal Encoding framework for Quadratic Unconstrained Binary Optimization (QUBO) problems, mapping $n_c$ classical variables onto $\mathcal{O}(\log n_c)$ qubits. This has been applied to both financial transaction settlement~\cite{huber2024exponential} and vehicle routing problems~\cite{leonidas2024qubit}, but this approach is essentially a lossy compression. It limits global entanglement, leading to reduced accuracy for complex problems and increased measurement overhead. To improve this, Podobrii \textit{et al.}~\cite{podobrii2025qubit} shifted the focus from global to local search, encoding only the neighborhood of a solution to handle larger instances. Additionally, Kanatbekova \textit{et al.}~\cite{kanatbekova2025qubit} further reduced auxiliary qubits for inequality constraints through exponential penalty functions. However, these methods remain heuristic and cannot guarantee mathematical equivalence to the original system dynamics, while also requiring higher precision or prior knowledge in practical implementations. Secondly, the problem partitioning strategy decomposes large-scale problems into smaller sub-tasks with fewer qubits. Shaydulin \textit{et al.}~\cite{zhou2023qaoa} proposed the QAOA-in-QAOA framework, which recursively decomposes large problem instances and reconstructs global solutions from subgraph results. Li \textit{et al.}~\cite{li2022large} later applied partitioning tools like measurement distribution reconstruction (MDR) to split Max-Cut problems into independent tasks, utilizing compensation mechanisms to handle boundary information loss. More recently, Hong \textit{et al.}~\cite{hong2025qubit} formulated this strategy within an Alternating Direction Method of Multipliers (ADMM)-based framework, enabling large-scale constrained optimization through alternating quantum sub-problems. However, by partitioning the global structure of the original problem, these approaches introduce significant classical communication overhead and cannot guarantee strict mathematical equivalence between the resultant quantum evolution and that of the original system. Thirdly, circuit-level compression optimizes physical implementations to reduce qubit numbers. Gokhale \textit{et al.}~\cite{gokhale2019asymptotic} showed that high-dimensional quantum states (Qutrits) can break depth limits for multi-controlled gates, achieving asymptotic circuit compression. Meanwhile, DeCross \textit{et al.}~\cite{decross2023qubit} utilized mid-circuit measurement and reset to enable qubit reuse, allowing fewer physical qubits to simulate larger systems over time. While these methods reduce qubit requirements, they rely on additional quantum operations (e.g., mid-circuit measurements for qubit reuse). These operations introduce additional noise, which degrades computational accuracy on practical quantum devices. Table~\ref{tab:comparison} compares our method with representative qubit reduction methods for QAOA. Although these existing methods reduce qubit requirements, but at the cost of degrading computational performance. Only the proposed EQE-QAOA satisfies all these features while preserving optimization accuracy.

\renewcommand{\arraystretch}{1.15}
\begin{table*}[t]
\caption{Comparison of representative qubit reduction methods for QAOA}
\label{tab:comparison}
\centering
\begin{tabular}{|p{5.2cm}|c|c|c|c|c|c|c|c|c|}
\hline
\textbf{Features} & [19] & [20] & [21] & [13] & [22] & [23] & [24] & [25] & \textbf{EQE-QAOA} \\
\hline
Preserves variable correlations &  &  & $\checkmark$ &  & $\checkmark$ &  & $\checkmark$ & $\checkmark$ & $\checkmark$ \\
\hline
Preserves problem structure & $\checkmark$ & $\checkmark$ & $\checkmark$ &  &  &  & $\checkmark$ & $\checkmark$ & $\checkmark$ \\
\hline
Preserves global search capability & $\checkmark$ &  & $\checkmark$ &  &  &  & $\checkmark$ & $\checkmark$ & $\checkmark$ \\
\hline
No additional noise-inducing operations & $\checkmark$ & $\checkmark$ & $\checkmark$ & $\checkmark$ & $\checkmark$ & $\checkmark$ &  &  & $\checkmark$ \\
\hline
No manual intervention required &  &  &  & $\checkmark$ & $\checkmark$ & $\checkmark$ & $\checkmark$ & $\checkmark$ & $\checkmark$ \\
\hline
Equivalence to the original QAOA evolution &  &  &  &  &  &  & $\checkmark$ & $\checkmark$ & $\checkmark$ \\
\hline
\textbf{Preserves optimization accuracy} &  &  &  &  &  &  &  &  & $\checkmark$ \\
\hline
\end{tabular}
\end{table*}

To overcome the performance degradation encountered by existing qubit reduction methods, we investigate the intrinsic physical structure of quantum dynamics to achieve more effective qubit reduction and lossless optimization. In quantum computing, the Hilbert space represents the complete state space of a quantum system, whose dimension grows exponentially with the number of qubits~\cite{nielsen2010quantum}. Extensive studies in quantum physics indicate that physical systems often possess inherent structural constraints, such as symmetries and conservation laws, which govern the decomposition of the Hilbert space into smaller invariant subspaces~\cite{weyl1950theory, tinkham2003group, nakazato2023invariant}. In particular, a symmetry of the system Hamiltonian implies a conserved quantity that remains unchanged during the dynamical evolution~\cite{salmistraro1990quantum}. This property enables a rigorous decomposition of the full Hilbert space into a direct sum of independent, dynamically invariant subspaces, where each subspace corresponds to an eigenspace of the conserved quantity. Magalhães de Castro \textit{et al.}~\cite{de2023analytically} further demonstrated that even highly structured and complex systems can enable a decomposition into lower-dimensional invariant subspaces. In such representations, the global evolution can be expressed as a collection of independent dynamical processes acting on these subspaces~\cite{hioe1981n}. These observations provide a fundamental theoretical foundation for our proposed qubit reduction framework design.

Inspired by an observation in quantum physics~\cite{nielsen2010quantum, weyl1950theory, tinkham2003group, nakazato2023invariant, salmistraro1990quantum, de2023analytically, hioe1981n}, we first prove that if the operators in QAOA share common symmetries, the system dynamics are strictly confined to a corresponding invariant subspace. We obtain the insight that there is no dynamical coupling between different subspaces, and this restriction is mathematically exact and lossless. This insight indicates that quantum algorithms can operate entirely within the invariant subspace, thereby avoiding the exponential complexity of the full-scale Hilbert space. Based on this, we propose a mapping to re-encode the full-scale Hilbert space into a reduced qubit space. Building on this principle, we propose an Equivalence-preserving Qubit Efficient QAOA (EQE-QAOA) framework, which exploits the underlying invariant subspace structures and constructs an isometric mapping into a reduced qubit space. Fig.~\ref{fig_0} illustrates the key idea of the proposed EQE-QAOA framework. This framework achieves qubit reduction, while attaining a performance equivalent to that of the original QAOA. The main contributions of this paper are summarized as follows:

\begin{figure}[t]
  \centering
  \includegraphics[width=0.9\linewidth]{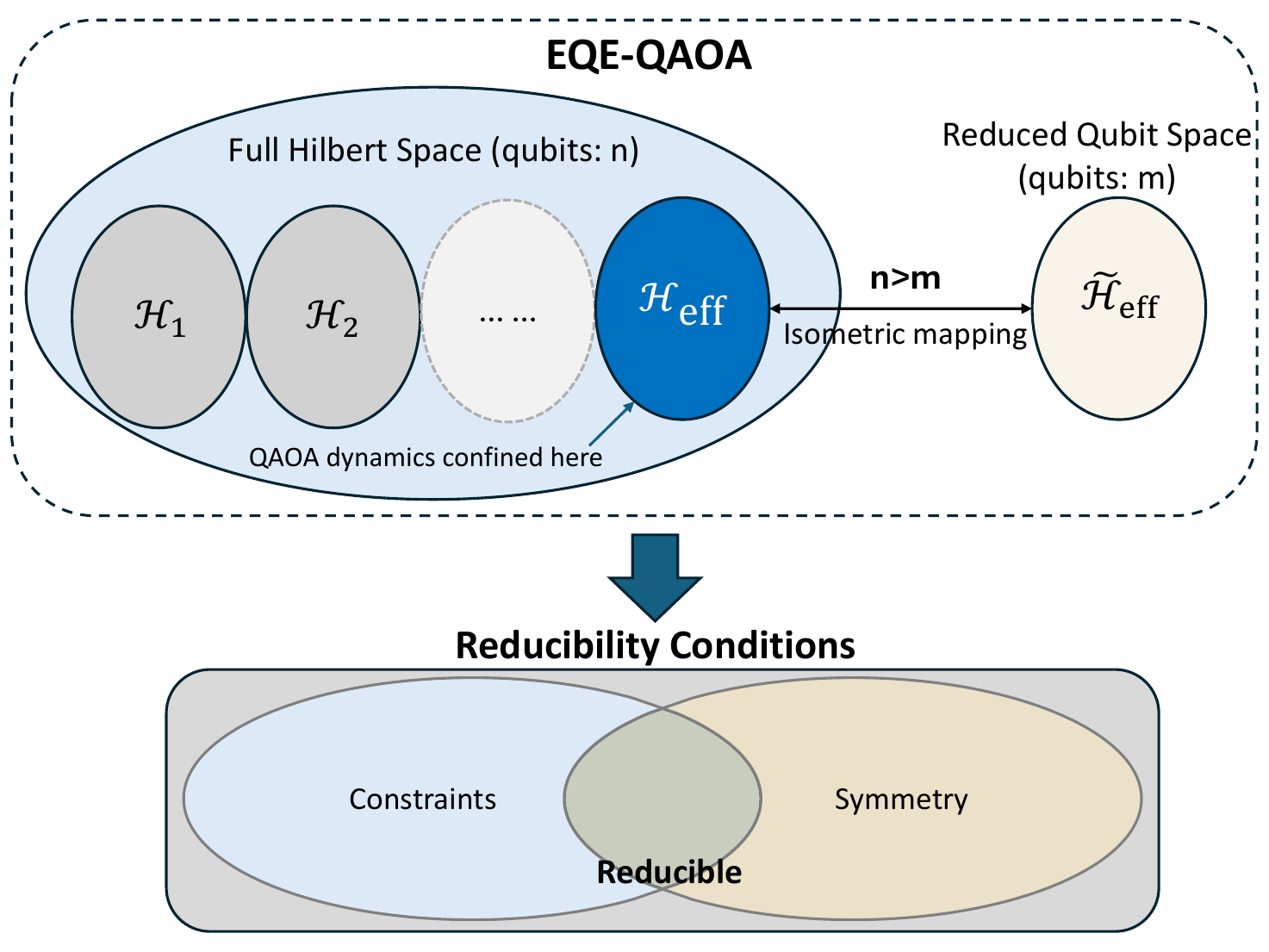}
  \caption{Overview of the proposed EQE-QAOA framework and its reducibility conditions.}
  \label{fig_0}
\end{figure}

\begin{enumerate}
    \item \textbf{Invariant subspace decomposition:} By utilizing the dynamical Lie algebra~\cite{d2021introduction} and calculating its commutant, we identify conserved quantities in the QAOA evolution. This enables the decomposition of the Hilbert space into invariant subspaces, with the evolution confined to a single subspace. Consequently, the space of QAOA evolution is reduced.
    \item \textbf{Proof of performance equivalence between subspace and full-space based QAOA:} We derive an equivalence theorem, guaranteeing that evolution in the subspace preserves identical quantum states and measurement statistics in the full-space evolution. This ensures that the evolution space can be reduced without any loss of information or degradation in optimization performance.
    \item \textbf{Isometric mapping to a lower-dimensional representation with fewer qubits:} We construct an isometric mapping that encodes the invariant subspace into a lower-dimensional qubit representation. This mapping allows the number of qubits to be reduced without information loss.
    \item \textbf{Deriving the conditions for EQE-QAOA applicability:} We derive the formal conditions under which a problem is suitable for our proposed EQE-QAOA framework. Our analysis shows that EQE-QAOA is broadly applicable to large-scale constrained optimization problems in engineering, where it reduces both the qubit resources and the computational complexity.
\end{enumerate}

The remainder of this paper is organized as follows. Section II provides the quantum representation and preliminaries. Section III introduces the EQE-QAOA framework and its qubit-efficient encoding scheme. Section IV presents the applicability and scalability of the proposed method through structural and complexity analysis. Section V presents numerical results, while Section VI concludes the paper.

\section{Quantum Representation and Preliminaries}

We first define a unified representation to express the combinatorial optimization problem in quantum form, specifying the problem Hamiltonian and the corresponding QAOA dynamics as the foundation of EQE-QAOA.

\subsection{From Combinatorial Optimization to Quantum Representation}

Many combinatorial optimization problems can be formulated as binary optimization tasks with decision variables $x = (x_1, x_2, \ldots, x_n)$, where $x_i \in \{0,1\}$ and $n$ denotes the number of decision variables. Within the QAOA framework, such problems are typically expressed as a QUBO model, defined by:%
\begin{equation}
    \min_{x \in \{0,1\}^n} f(x) = x^\top W x + c^\top x + c_0 + \sum_{k} \lambda_k g_k(x),
\end{equation}

\noindent  where $W \in \mathbb{R}^{n\times n}$ is a symmetric matrix, $c \in \mathbb{R}^n$ is a bias vector, and $c_0$ is a constant offset. The terms $\sum_k \lambda_k g_k(x)$ represent penalty functions used for enforcing problem constraints. By incorporating these penalties into the objective, constrained optimization problems can be reformulated as unconstrained ones, which enables a direct encoding into the QAOA cost Hamiltonian.

To embed the classical objective into a quantum representation, each binary variable $x_i$ is mapped to a qubit operator $\hat{x}_i$ through the transformation:%
\begin{equation}
    x_i \mapsto \hat{x}_i = \frac{I - Z_i}{2},
\end{equation}

\noindent where $I$ denotes the identity operator and $Z_i$ represents the Pauli-$Z$ operator acting on the $i$-th qubit, and $i \in \{1,2,\ldots,n\}$. An operator is defined as a linear map on the Hilbert space, with $\{I, X, Y, Z\}$ as the basic single-qubit operators. Under this mapping, each binary string $x \in \{0,1\}^n$ corresponds to a computational basis state $|x\rangle$ in the Hilbert space $\mathcal{H} = (\mathbb{C}^2)^{\otimes n}$, allowing the classical objective to be encoded into the quantum Hamiltonian representation. The main notations used throughout this paper are summarized in Table~\ref{tab:notations}.

\renewcommand{\arraystretch}{1.15}
\begin{table}[t]
\caption{Key Notations Used Throughout the Paper}
\label{tab:notations}
\centering
\begin{tabular}{ll}
\hline
\textbf{Symbol} & \textbf{Meaning} \\
\hline
$n$ & Number of qubits in the original formulation \\
$m$ & Number of qubits after reduction \\
$\mathcal H=(\mathbb C^2)^{\otimes n}$ & Original Hilbert space \\
$\mathcal H_{\mathrm{eff}}$ & Effective invariant subspace \\
$M=\dim(\mathcal H_{\mathrm{eff}})$ & Dimension of effective subspace \\
$\tilde{\mathcal H}_{\mathrm{eff}}=(\mathbb C^2)^{\otimes m}$ & Reduced Hilbert space \\
$H_C, H_M$ & Cost and mixer Hamiltonians \\
$H_C^{\mathrm{eff}}, H_M^{\mathrm{eff}}$ & Hamiltonians restricted to $\mathcal H_{\mathrm{eff}}$ \\
$\tilde H_C, \tilde H_M$ & Induced Hamiltonians in $\tilde{\mathcal H}$ \\
$\mathfrak g = \mathrm{Lie}(\mathbf{i}H_C,\mathbf{i}H_M)$ & Dynamical Lie algebra \\
$\mathfrak g'$ & Commutant of $\mathfrak g$ \\
$Q$ & Conserved quantity satisfying $[Q,H_{C,M}]=0$ \\
$\Pi_{\mathrm{eff}}$ & Projector onto $\mathcal H_{\mathrm{eff}}$ \\
$V$ & Isometric embedding $V:\tilde{\mathcal H} \to \mathcal H_{\mathrm{eff}}$ \\
$U(\theta)$ & Full-space QAOA unitary \\
$\tilde U(\theta)$ & Reduced-space QAOA unitary \\
\hline
\end{tabular}
\end{table}

\subsection{Unified Form of the Quantum Hamiltonian}

In the QAOA framework, the optimization process is driven by a pair of core Hamiltonians, namely the Cost Hamiltonian $H_C$ and the Mixer Hamiltonian $H_M$. $H_C$ encodes the objective function into the energy landscape, and $H_M$ enables transitions between computational basis states to explore the solution space.

The classical objective $f(x)$ is transformed into a cost Hamiltonian $H_C$ by applying the mapping in Eq.~(2) to each variable. For QUBO problems, $H_C$ is diagonal in the computational basis, ensuring a direct spectral correspondence between the eigenvalues of $H_C$ and the objective values:%
\begin{equation}
    H_C |x\rangle = f(x) |x\rangle, \quad \forall x \in \left \{ 0,1 \right \}^n. 
\end{equation}

Since $H_C$ is diagonal in the computational basis, transitions between basis states are defined as:%
\begin{equation}
    H_M = \sum_{\ell} \beta_\ell P_\ell,
\end{equation}

\noindent where $\ell \in \{1,\dots,L\}$ indexes the terms in the $H_M$, and $\beta_\ell \in \mathbb{R}$ is a real coefficient controlling the strength of each term. $P_\ell \in \{I, X, Y, Z\}^{\otimes n}$ is a Pauli string containing at least one non-diagonal operator. Each term $P_\ell$ enables transitions between computational basis states. For any two computational basis states $|x\rangle$ and $|x'\rangle$, $H_M$ facilitates transitions of the form:%
\begin{equation}
    \langle x' | H_M | x \rangle \neq 0 \quad \Rightarrow \quad x' = x \oplus \delta,
\end{equation}

\noindent where $\delta \in \{0,1\}^n$ specifies a bit-flip pattern. For instance, the most common choice is the transverse-field mixer $H_M = \sum_i X_i$, where each term $X_i$ allows the system to traverse the Hamming space by flipping individual bits ($\delta$ being a unit vector).

\subsection{Quantum Optimization Dynamics}

The QAOA dynamics are generated by $H_C$ and $H_M$. The parameterized QAOA state with $P$ layers is defined as:%
\begin{equation}
    |\psi(\boldsymbol{\gamma}, \boldsymbol{\beta})\rangle = \left( \prod_{p=1}^{P} e^{-i \beta_p H_M} e^{-i \gamma_p H_C} \right) |\psi_0\rangle,
\end{equation}

\noindent where $|\psi_0\rangle$ denotes the initial state (typically the ground state of $H_M$, such as the uniform superposition). The evolution is governed by two sets of variational parameters: $\boldsymbol{\gamma} = (\gamma_1, \dots, \gamma_P)$ and $\boldsymbol{\beta} = (\beta_1, \dots, \beta_P)$. In each of the $P$ layers, the operator $e^{-i \gamma_p H_C}$ assigns phases to each computational basis state depending on the $f(x)$. Subsequently, the operator $e^{-i \beta_p H_M}$ mixes the states by enabling the legitimate transitions between them. By iteratively applying these layers, the algorithm explores the Hilbert space to concentrate the probability amplitude on the specific states corresponding to the optimal solutions of the original QUBO problem.

\section{The Proposed EQE-QAOA}

To improve the qubit efficiency of QAOA, we propose EQE-QAOA. We first analyze its quantum dynamics and prove that the evolution is confined to an invariant subspace. However, the algorithm is still implemented in the original space with the same number of qubits. To address this, we construct an isometric mapping that encodes the invariant subspace using fewer qubits, enabling qubit reduction while preserving the original QAOA performance.

\subsection{Identification of the Invariant Subspace}

We consider a quantum dynamical system defined on the Hilbert space $\mathcal{H} = (\mathbb{C}^2)^{\otimes n}$, which is the state space containing all possible quantum states of the $n$-qubit system. As implied by Eq.~(6), the dynamics within this space are generated by $H_C$ and $H_M$, which together define the fundamental operators governing the evolution. Therefore, the set of all possible unitary operations, $\mathfrak{g}$, is determined by the dynamical Lie algebra:%
\begin{equation}
    \mathfrak{g} = \mathrm{Lie}(iH_C, iH_M),
\end{equation}

\noindent where $i = \sqrt{-1}$ denotes the imaginary unit. In practice, $\mathfrak g$ is obtained by repeatedly applying commutators of the generators $\{iH_C,iH_M\}$ until closure under the Lie bracket $[\cdot,\cdot]$ is reached. The structure of $\mathfrak{g}$ fully determines the system’s reachability. The dimension and representation of $\mathfrak{g}$ dictate whether the dynamics induced by the Hamiltonians are capable of exploring the full Hilbert space, effectively defining the expressive power of the QAOA.

Having established the reachability of the system via $\mathfrak{g}$, we now define the commutant of $\mathfrak{g}$ as $\mathfrak{g}'$ to capture the system symmetry:%
\begin{equation}
    \mathfrak g' = \{ Q \in \mathsf{End}(\mathcal H) : [Q,X]=0,\ \forall X\in\mathfrak g \},
\end{equation}

\noindent where $\mathfrak{g}'$ is the set of operators in the endomorphism space $\mathsf{End}(\mathcal{H})$ that satisfy $[Q,X]=0$ for all $X \in \mathfrak{g}$. These operators $Q$ commute with all generators in $\mathfrak{g}$, meaning that applying $Q$ before or after the dynamics produces the same result, and the two operations are independent. Therefore, these operators $Q$ leave the QAOA dynamics invariant and correspond to symmetries of the system. As a result, they induce a decomposition of the Hilbert space into invariant subspaces. A fundamental property of Lie algebras ensures that an operator $Q$ commutes with every element of $\mathfrak{g}$ if and only if it commutes with its generators $H_C$ and $H_M$~\cite{weyl1950theory}. Therefore, we simplify the search for $\mathfrak{g}'$ as:%
\begin{equation}
    \mathfrak g' = \{ Q : [Q,H_C]=0,\ [Q,H_M]=0 \}.
\end{equation}

\noindent To construct operators in the commutant $\mathfrak{g}'$, we expand an unknown operator $Q$ in the $n$-qubit Pauli basis $\mathcal P_n$ as:%
\begin{equation}
    Q \;=\; \sum_{P\in\mathcal P_n} q_P\, P.
\end{equation}

\noindent where $\{q_P\}$ are coefficients associated with each Pauli string $P$, which can be determined by solving the commutation constraints $[Q,H_C]=0$ and $[Q,H_M]=0$. Therefore, all operators $Q$ satisfying these constraints constitute the commutant $\mathfrak g'$.

Then we utilize this conserved quantity $Q$ to decompose the Hilbert space into invariant subspaces. For a single QAOA layer ($P=1$ in Eq.~(6)), the unitary operator is given by $U(\gamma,\beta)=e^{-i\beta H_M}e^{-i\gamma H_C}$, which is generated by exponentials of $H_C$ and $H_M$. From Eq.~(9), we have $[Q,H_C]=0$ and $[Q,H_M]=0$. Therefore, $Q$ also commutes with the unitary operator $U(\gamma,\beta)$. This relationship can be written as:%
\begin{equation}
    U(\gamma, \beta) Q U^\dagger(\gamma, \beta) = e^{-i\beta H_M} e^{-i\gamma H_C} Q e^{i\gamma H_C} e^{i\beta H_M} = Q.
\end{equation}

\noindent This commutation relationship implies that the eigenstructure of $Q$ is preserved under the QAOA evolution. Therefore, any nontrivial operator $Q \in \mathfrak{g}'$ represents a conserved quantity. In particular, if a state $|\psi\rangle$ is an eigenvector of $Q$ with eigenvalue $\lambda_i$, i.e., $Q |\psi\rangle = \lambda_i |\psi\rangle$, then the evolved state $U|\psi\rangle$ satisfies:%
\begin{equation}
    Q (U|\psi\rangle) = U Q |\psi\rangle = \lambda_i (U|\psi\rangle),
\end{equation}

\noindent showing that $U|\psi\rangle$ remains in the same eigenspace associated with $\lambda_i$. Therefore, the eigenvalues $\{\lambda_i\}$ of the conserved quantity $Q$ are conserved throughout the evolution, and quantum states cannot transition between different eigenspaces labeled by these eigenvalues. The joint eigenspaces of $Q$ thus define invariant subspaces of the Hilbert space. Consequently, the full Hilbert space admits the following decomposition:%
\begin{equation}
    \mathcal H \;=\; \bigoplus_{\alpha} \mathcal H_{\alpha},
\end{equation}

\noindent where each subspace $\mathcal{H}_\alpha$ corresponds to a specific eigenvalue of the conserved quantity $Q$, and $\alpha$ indexes the distinct eigenvalues of $Q$. 

To explicitly represent each invariant subspace and show that the dynamics is confined within it, this decomposition can equivalently be expressed through a set of orthogonal projections $\{\Pi_\alpha\}$ onto the invariant subspaces:%
\begin{equation}
    [\Pi_{\alpha}, H_C] = 0, \quad [\Pi_{\alpha}, H_M] = 0,
\end{equation}

\noindent satisfying the completeness and orthogonality relationships $\sum_{\alpha} \Pi_{\alpha} = I$ and $\Pi_{\alpha} \Pi_{\alpha'} = \delta_{\alpha\alpha'} \Pi_{\alpha}$ ($\delta_{\alpha\alpha'}$ is the Kronecker delta). These relationships ensure that the dynamics never induce transitions between different subspaces:%
\begin{equation}
    \Pi_{\alpha'} U(\boldsymbol{\gamma}, \boldsymbol{\beta}) \Pi_{\alpha} = 0 \quad \text{for } \alpha \neq \alpha'.
\end{equation}

\noindent where $U(\boldsymbol{\gamma}, \boldsymbol{\beta}) = \prod_{p=1}^P e^{-i\beta_p H_M} e^{-i\gamma_p H_C}$. This shows that the evolution is strictly confined to the invariant subspace that contains the initial state $|\psi_0\rangle$. Specifically, if $|\psi_0\rangle \in \mathcal{H}_0$, then for any layer $P$ and variational parameters $\{\boldsymbol{\gamma}, \boldsymbol{\beta}\}$, the evolved state $|\psi(\boldsymbol{\gamma}, \boldsymbol{\beta})\rangle$ remains within $\mathcal{H}_0$. We therefore define $\mathcal{H}_{\text{eff}} \triangleq \mathcal{H}_0$ as the effective invariant subspace. This subspace $\mathcal{H}_{\text{eff}}$, rather than the exponentially larger full Hilbert space $\mathcal{H}$, constitutes the minimal state space required for the algorithm.

\subsection{Invariant Subspace Representation of QAOA}

To derive the dynamics within $\mathcal{H}_{\text{eff}}$, we define the effective Hamiltonians as $H_C^{\text{eff}}$ and $H_M^{\text{eff}}$. They can be represented on the original Hilbert space $\mathcal{H}$ via projection onto the effective subspace $\mathcal{H}_{\text{eff}}$:%
\begin{equation}
    H_C^{\text{eff}} \triangleq \Pi_{\text{eff}} H_C \Pi_{\text{eff}}, \quad H_M^{\text{eff}} \triangleq \Pi_{\text{eff}} H_M \Pi_{\text{eff}}.
\end{equation}

\noindent Observe from Eq.~(14) that $[H_C,\Pi_{\text{eff}}]=0$ and $[H_M,\Pi_{\text{eff}}]=0$, these projected operators act identically to the original Hamiltonians within $\mathcal{H}_{\text{eff}}$, while being zero on all other subspaces. The full-space evolution is governed by the unitary operator $U(\boldsymbol{\gamma}, \boldsymbol{\beta})$, from which we define a subspace unitary $U_{\text{eff}}(\boldsymbol{\gamma}, \boldsymbol{\beta})$ acting on $\mathcal{H}_{\text{eff}}$:%
\begin{equation}
    U_{\text{eff}}(\boldsymbol{\gamma}, \boldsymbol{\beta}) \triangleq \prod_{p=1}^P e^{-i\beta_p H_M^{\text{eff}}} e^{-i\gamma_p H_C^{\text{eff}}}.
\end{equation}

To ensure that the QAOA dynamics remains confined within the effective subspace, we first prove the invariance of the projector.

\begin{lemma}
    Let $\Pi$ be an orthogonal projector satisfying $[\Pi,H_C]=0$ and $[\Pi,H_M]=0$. Then for any parameters $\gamma,\beta$, the projector commutes with the full QAOA unitary evolution:%
    \begin{equation}
        [\Pi, U(\boldsymbol{\gamma}, \boldsymbol{\beta})] = 0.
    \end{equation}
\end{lemma}

The detailed proof is provided in Appendix~A.

Based on the projector invariance in Lemma 1, the following theorem shows the equivalence between the full-space and reduced-space dynamics.

\begin{theorem}
    Assume the initial state $|\psi_0\rangle$ is contained within the effective subspace, i.e., $\Pi_{\mathrm{eff}}|\psi_0\rangle = |\psi_0\rangle$. Let $|\psi\rangle = U|\psi_0\rangle \in \mathcal{H}$ be the full-space evolved state and $|\psi_{\mathrm{eff}}\rangle = U_{\mathrm{eff}}|\psi_0\rangle \in \mathcal{H}_{\mathrm{eff}}$ be the reduced-space evolved state. Then:

    \begin{enumerate}
        \item \textbf{State equivalence:} The state trajectory is confined to the effective subspace, $|\psi\rangle \in \mathcal{H}_{\mathrm{eff}}$, and $|\psi\rangle = |\psi_{\mathrm{eff}}\rangle$.

        \item \textbf{Observable equivalence:} For any observable $O$ that commutes with the subspace projector ($[O, \Pi_{\mathrm{eff}}] = 0$), the expectation values of $O$ satisfy:%
            \begin{equation}
                \langle\psi|O|\psi\rangle = \langle\psi_{\mathrm{eff}}| O|_{\mathcal{H}_{\mathrm{eff}}} |\psi_{\mathrm{eff}}\rangle.
            \end{equation}
            
        \item \textbf{Measurement equivalence:} For any solution $x \in \{0,1\}^n$, the probability distributions are identical:%
        \begin{equation}
            \Pr_{\mathrm{full}}(x) = \Pr_{\mathrm{eff}}(x).
        \end{equation}

        \noindent Specifically, if the basis state $|x\rangle$ is orthogonal to $\mathcal{H}_{\mathrm{eff}}$, then we have $\Pr_{\mathrm{full}}(x) = 0$.
    \end{enumerate}
\end{theorem}

The detailed proof is provided in Appendix~B.

Theorem 1 rigorously proves the equivalence between the original system $(\mathcal{H}, H_C, H_M, |\psi_0\rangle)$ and the minimal effective system $(\mathcal{H}_{\mathrm{eff}}, H_C^{\mathrm{eff}}, H_M^{\mathrm{eff}}, |\psi_0\rangle)$. It demonstrates that this reduction is exact, meaning that the QAOA in subspace reproduces the full QAOA's dynamics and statistics. Consequently, the QAOA dynamics can be fully represented within $\mathcal{H}_{\mathrm{eff}}$, which defines the minimal Hilbert space of QAOA.

\subsection{Qubit Reduction via Isometric Mapping}

Since the QAOA dynamics are confined within the effective invariant subspace $\mathcal{H}_{\mathrm{eff}}$, it is unnecessary to represent the full Hilbert space. To reduce the number of qubits, we construct an isometric mapping to encode $\mathcal{H}_{\mathrm{eff}}$ into a reduced qubit space. The dimension of $\mathcal{H}_{\mathrm{eff}}$ is denoted by $M \triangleq \dim(\mathcal{H}_{\mathrm{eff}})$, which is the number of basis states needed to represent the $\mathcal{H}_{\mathrm{eff}}$. To construct the isometric mapping, we choose an orthonormal basis for $\mathcal{H}_{\mathrm{eff}}$, denoted by $\{|\phi_k\rangle\}_{k=0}^{M-1}$. The effective subspace can therefore be written as $\mathcal{H}_{\mathrm{eff}}
= \mathrm{span}\{|\phi_k\rangle\}_{k=0}^{M-1}$. To represent $\mathcal{H}_{\mathrm{eff}}$ with a reduced number of qubits, we determine the minimal number of qubits required. Since a quantum system having $m$ qubits can represent a Hilbert space of dimension $2^m$, the minimal number of qubits required for representing $\mathcal{H}_{\mathrm{eff}}$ is determined by $2^m \ge M$, leading to $m \triangleq \lceil \log_2 M \rceil$. In this reduced representation, each orthonormal basis vector $|\phi_k\rangle$ of the physical subspace $\mathcal{H}_{\mathrm{eff}}$ is mapped to a corresponding orthonormal computational basis state $|\tilde{\phi} _k\rangle$ of $m$-qubit. These states span the effective reduced-qubit space:%
\begin{equation}
    \tilde{\mathcal H}_{\mathrm{eff}}=
\mathrm{span}\{|\tilde{\phi}_k\rangle\}_{k=0}^{M-1}
\subseteq
\tilde{\mathcal H}.
\end{equation}

To formally bridge these two spaces, we define an isometric mapping $V : \tilde{\mathcal H}_{\mathrm{eff}} \to \mathcal H_{\mathrm{eff}}$. This operator maps states in $\tilde{\mathcal H}_{\mathrm{eff}}$ to the original subspace $\mathcal H_{\mathrm{eff}}$, while $V^\dagger$ performs the reverse mapping. Given the orthonormal bases $\{|\phi_k\rangle\}_{k=0}^{M-1}$ for $\mathcal H_{\mathrm{eff}}$ and $\{|\tilde{\phi}_k\rangle\}_{k=0}^{M-1}$ for $\tilde{\mathcal H}_{\mathrm{eff}}$, the isometry is defined as:%
\begin{equation}
    V \triangleq \sum_{k=0}^{M-1} |\phi_k\rangle \langle \tilde{\phi}_k|.
\end{equation}

The following lemma summarizes the properties of the mapping, which enable the reconstruction of QAOA dynamics in the reduced-qubit space.

\begin{lemma}
    The operator $V$ satisfies the following identities:%
    \begin{equation}
       V^\dagger V = I_{\tilde{\mathcal{H}}_{\mathrm{eff}}}, \quad VV^\dagger = \Pi_{\mathrm{eff}}.
    \end{equation}
\end{lemma}

The detailed proof is provided in Appendix~C.

Given the isometry $V$, we map the restricted Hamiltonians from the $n$-qubit space to the minimal $m$-qubit space. We define the induced Hamiltonians on $\tilde{\mathcal{H}}_{\mathrm{eff}}$ as:%
\begin{equation}
    \tilde{H}_C^{\mathrm{eff}} \triangleq V^\dagger H_C V, \quad \tilde{H}_M^{\mathrm{eff}} \triangleq V^\dagger H_M V.
\end{equation}

\noindent The Hamiltonians on $\tilde{\mathcal H}_{\mathrm{eff}}$ behave exactly the same as the original Hamiltonians within the effective subspace. In the following, we prove that the evolution induced by these Hamiltonians, when mapped via $V$, exactly reproduces the original QAOA trajectory.

\begin{theorem}
    For any parameters $\gamma, \beta \in \mathbb{R}$ and the isometric mapping $V$, the following operator identities hold:%
    \begin{equation}
        e^{-i\gamma H_C} V = V e^{-i\gamma \tilde{H}_C^{\mathrm{eff}}}, \quad e^{-i\beta H_M} V = V e^{-i\beta \tilde{H}_M^{\mathrm{eff}}}.
    \end{equation}

    \noindent Consequently, for a QAOA circuit with $P$ layers and parameters $\{\boldsymbol{\gamma}, \boldsymbol{\beta}\}$, the total evolution operators $U(\boldsymbol{\boldsymbol{\gamma}, \boldsymbol{\beta}})$ and $\tilde{U}(\boldsymbol{\gamma}, \boldsymbol{\beta})$ satisfy:%
    \begin{equation}
        U(\boldsymbol{\gamma}, \boldsymbol{\beta}) V = V \tilde{U}(\boldsymbol{\gamma}, \boldsymbol{\beta}).
    \end{equation}
\end{theorem}

The detailed proof is provided in Appendix~D.

\begin{corollary}
    Based on the Theorem~1, for any reduced initial state $|\tilde\psi_0\rangle$ with $|\psi_0\rangle = V|\tilde\psi_0\rangle$, the dynamics generated in the reduced $m$-qubit space are equivalent to the original QAOA dynamics in $\mathcal H$. Consequently, the reduced system reproduces the same state trajectory, expectation values, and measurement statistics as the full system.
\end{corollary}

The detailed proof is provided in Appendix~E.

Since the measurement statistics are preserved exactly, the quantum evolution can be executed entirely within the $m$-qubit space by replacing the original Hamiltonians with their induced counterparts:%
\begin{equation}
    |\tilde\psi(\boldsymbol{\gamma}, \boldsymbol{\beta})\rangle = \prod_{p=1}^{P} e^{-i\beta_p \tilde H_M^{\mathrm{eff}}} e^{-i\gamma_p \tilde H_C^{\mathrm{eff}}} |\tilde\psi_0\rangle.
\end{equation}

\noindent Under this scheme, the quantum state in the original problem space is recovered, as $|\psi(\boldsymbol{\gamma}, \boldsymbol{\beta})\rangle = V|\tilde\psi(\boldsymbol{\gamma}, \boldsymbol{\beta})\rangle$. This allows the QAOA evolution to be executed entirely in the reduced space, while preserving the original optimization behavior.

\section{Applicability and Scalability Analysis of EQE-QAOA}

To evaluate the applicability of the proposed EQE-QAOA, we first derive the structural limits under which qubit reduction is not possible. Based on this, we then derive the conditions under which EQE-QAOA becomes applicable. We further analyze both the qubit reduction and computational complexity to quantify qubit savings and speedup. Finally, we show that EQE-QAOA imposes a much lower computational complexity than the original QAOA.

\subsection{Problem Conditions for Qubit Reducibility}

\subsubsection{Irreducible Problem Structures}

We derive the conditions under which the problem is irreducible, which occurs when it is \textit{unconstrained} and all variables are \textit{independent}. In such cases, no variable can be represented as a function of others (e.g., $x_1 = x_2$ or $x_3 = x_1 \oplus x_2$). As a result, no solutions can be merged, and every solution in the space must be evaluated individually. Consequently, the QAOA evolution spans the entire solution space. This indicates that it cannot be decomposed into smaller invariant subspaces. From a quantum dynamical perspective, this irreducible structure implies that the QAOA evolution does not preserve any nontrivial symmetry. In other words, the only operators that commute with both $H_C$ and $H_M$ are scalar multiples of the identity, so the commutant of the QAOA Hamiltonians is given by:%
\begin{equation}
    \{A \in \mathrm{End}(\mathcal H) : [A,H_C]=[A,H_M]=0\} = \mathbb C I.
\end{equation}

\noindent In this case, only trivial conserved quantities exist, and the Hilbert space admits no invariant decomposition. Consequently, no compressible redundancy is present, and qubit reduction is not possible.

\subsubsection{Reducible Problem Structures}

Based on the above irreducibility condition, we now derive the corresponding condition for qubit reduction. In the irreducible case shown in Eq.~(28), the commutant of the QAOA generators is trivial, it contains only scalar multiples of the identity. Therefore, the reducible case is derived by the opposite condition, where the commutant of the QAOA generators is nontrivial:%
\begin{equation}
    \{A : [A,H_C]=[A,H_M]=0\} \neq \mathbb C I,
\end{equation}

\noindent which implies the existence of a nontrivial Hermitian operator that commutes with both $H_C$ and $H_M$. Such an operator acts as a conserved quantity and induces a decomposition of the Hilbert space into invariant subspaces. 

This reducible condition arises when redundancy exists in the solution space, implying that some solutions can be ruled out or merged, thereby restricting the problem to a smaller subspace. For optimization problems, this typically occurs under two conditions. Firstly, constraints introduce coupling among variables. For example, constraints such as $\sum_i x_i = k$ restrict the feasible solutions to a subspace. As a result, many states in the original $2^n$ Hilbert space become unreachable, leading to a form of spatial redundancy. In this case, the constraint can be represented by a Hermitian operator $Q$ that is preserved under the QAOA dynamics, inducing an invariant subspace. Secondly, symmetry introduces equivalence among solutions. If the objective function is invariant under a nontrivial permutation group $\mathcal S_f$ with $|\mathcal S_f| > 1$, multiple solutions become equivalent. This induces redundancy, allowing equivalent states to be merged into a reduced representation. In both cases, structure restricts the dynamics to a subspace of the Hilbert space, hence enabling qubit reduction. The problem conditions are summarized in Table~\ref{tab:structural_classification}.

\begin{table}[t]
\centering
\setlength{\tabcolsep}{5pt}
\renewcommand{\arraystretch}{1.2}
\caption{Problem Conditions of Qubit Reducibility.}
\label{tab:structural_classification}
\begin{tabular}{p{3.5cm} c c}
\hline
\textbf{Problem Condition} & \textbf{Algebraic Criterion} & \textbf{Reducibility} \\
\hline
Unconstrained \& independent variables & $\mathfrak{g}' = \mathbb{C}I$ & \textbf{Irreducible} \\
\hline
Constraints & $Q \in \mathfrak{g}' \setminus \mathbb{C}I$ & \textbf{Reducible} \\
Permutation symmetries & $|\mathcal{S}_f| > 1$ & \textbf{Reducible} \\
\hline
\end{tabular}
\end{table}

\medskip
Note that EQE-QAOA applies whenever the problem is not unconstrained and the variables are not independent. Our proposed EQE-QAOA can be applied to reduce the number of qubits, when constraints or symmetry are present, as invariant subspaces arise. In practical settings, such structures are very common in engineering optimization problems. Most real-world problems involve constraints or variable dependencies, such as resource allocation with budget or capacity limits, communication networks with repeated topology patterns, and scheduling or routing tasks with coupled decisions. These structures naturally induce invariant subspaces, making EQE-QAOA applicable to a wide range of practical problems. By exploiting this structure, EQE-QAOA reduces the number of required qubits without affecting the original optimization performance, enabling larger problems to be handled on near-term quantum devices and improving the efficiency of classical simulations.

\subsection{Resource and Complexity Analysis}

To evaluate the practical impact of EQE-QAOA, in Theorem~3, we derive the fundamental relationship between the invariant subspace dimension $M$ and the required number of qubits $m$. Corollary~2 further shows that the proposed EQE-QAOA requires fewer qubits.

\begin{theorem}
    Let $M = \dim(\mathcal{H}_{\mathrm{eff}})$. Any quantum system consisting of $m$ qubits has a total dimension of $2^m$. If the $\mathcal{H}_{\mathrm{eff}}$ is mapped into an $m$-qubit space $\tilde{\mathcal H}_{\mathrm{eff}}$ without losing information, the following must hold:%
    \begin{equation}
        2^m \ge M \quad\Longleftrightarrow\quad m \ge \lceil\log_2 M\rceil.
    \end{equation}

    \noindent Moreover, $m = \lceil\log_2 M\rceil$ qubits are sufficient to represent the system.
\end{theorem}

\begin{proof}
    A mapping $V$ that preserves the state (an isometry) must be injective, implying that it cannot map two different states to the same point. Mapping an $M$-dimensional space into an $m$-qubit space requires the target space to have at least $M$ dimensions to prevent information loss:%
    \begin{equation}
        2^m = \dim[(\mathbb{C}^2)^{\otimes m}] \ge M.
    \end{equation}

    \noindent Taking the base-2 logarithm of both sides results in $m \ge \log_2 M$. Since the number of qubits must be an integer, the requirement becomes $m \ge \lceil\log_2 M\rceil$.
\end{proof}

\begin{corollary}
    If the QAOA dynamics are confined to a lower-dimensional subspace $\mathcal{H}_{\mathrm{eff}}$ with $M = \dim(\mathcal{H}_{\mathrm{eff}}) < 2^n$, then, via an isometric mapping, the space $\tilde{\mathcal H}_{\mathrm{eff}}$ can be represented using $m$ qubits, where $m$ satisfies:%
    \begin{equation}
        m = \lceil\log_2 M\rceil \le n.
    \end{equation}

\end{corollary}

\begin{proof}
    By definition, the number of qubits required to represent a space of dimension $M$ is $m = \lceil \log_2 M \rceil$. Since $M < 2^n$, we have $\log_2 M < n$. Taking the ceiling on both sides gives $m = \lceil \log_2 M \rceil \le n$. 
\end{proof}

To evaluate the practical impact of EQE-QAOA, we consider the lower bound on qubit requirements for problems exhibiting high degrees of symmetry. For instance, in problems exhibiting complete permutation symmetry (e.g., for Max-Cut on all-to-all connected graphs where $M = n+1$), the number of qubits $m$ scale logarithmically with the original problem size $n$:%
\begin{equation}
    m_{\min} = \lceil \log_2 (n+1) \rceil.
\end{equation}

\noindent This represents the theoretical maximum compression achievable by the EQE-QAOA. In such cases, the qubit savings $\Delta = n - m_{\min}$ grow linearly with $n$.

The proposed EQE-QAOA has lower computational complexity due to the reduction of qubits. In the original QAOA, $n$ qubits are required to represent the quantum state, with a state space of dimension $2^n$. As a result, each layer has a computational complexity of $O(2^n)$. By contrast, in the proposed EQE-QAOA, only $m$ qubits are required to represent the quantum state, with a state space of dimension $2^m$. As a result, each QAOA layer has a computational complexity of $O(2^m)$. In the most extreme cases, only $m_{\min}$ qubits are required as shown in Eq.~(33), with a state space of dimension $n+1$. As a result, the compression ratio is $\eta = 2^n / (n+1)$, which grows exponentially with the problem size $n$. In this case, EQE-QAOA reduces the number of required qubits from $n$ to $m_{\min}$, while reducing the computational complexity from $O(2^n)$ to $O(n)$.

\section{Numerical Results}

\subsection{Experimental Settings}

All experiments are conducted using the Qiskit Aer statevector simulator on a classical computer equipped with an Intel Core i7-13700K CPU and 16 GB RAM. Due to the exponential scaling of the Hilbert space, classical simulation is primarily limited by memory resources. We consider the Max-Cut problem, which divides the vertices of a graph into two sets to maximize the number of edges between them. We evaluate both the original QAOA on $n$ qubits and the proposed EQE-QAOA on $m$ qubits. The number of QAOA layers is fixed to $P=2$ for all experiments. For each problem instance, we perform five independent optimization processes with different random initializations of the parameters $(\boldsymbol{\gamma}, \boldsymbol{\beta})$, and report the best result among them. Within each optimization process, the original QAOA and EQE-QAOA use identical parameters. Depending on the problem setting, the mixer $H_M$ is chosen either as the standard $X$-mixer or a constraint-preserving mixer.

We evaluate EQE-QAOA on two groups of problems. (i) The first group is designed to evaluate how problem symmetry affects the efficiency of EQE-QAOA. It consists of general graph families, including cycle graphs, complete graphs $K_n$ (full permutation symmetry, leading to the strongest reduction), and Erdős–Rényi random graphs (with edge probability $p_{\text{edge}} = 0.5$, where EQE-QAOA cannot effectively reduce the number of qubits). The problem sizes are chosen as $n \in \{6,8,10,12\}$. For each pair $(\text{family}, n)$, the graph structure is fixed, while parameter trials are generated using independent random seeds. (ii) The second group is designed to evaluate how constraints affect the efficiency of EQE-QAOA.
It considers constrained Max-Cut problems with a fixed-cardinality constraint on Erdős–Rényi random graphs. In this setting, we fix $n=12$ and impose a Hamming-weight constraint $\sum_{i=1}^{n} x_i = k$, where $k \in \{1,2,3,4,6\}$, which restricts feasible solutions to cuts with exactly $k$ vertices in one partition. Furthermore, an XY-type mixer is chosen as $H_M$ to preserve the Hamming-weight constraint. To validate EQE-QAOA, we evaluate three aspects: (i) qubit reduction, (ii) solution equivalence, and (iii) state representation memory usage. We validate the equivalence between the original QAOA and the proposed EQE-QAOA using three complementary metrics: the state fidelity $F(|\psi\rangle,|\tilde{\psi}\rangle)$ (reported as $F-1$), the absolute difference in the Max-Cut expectation value $|\Delta E| = |\langle H_C\rangle - \langle \tilde{H}_C\rangle|$, and the total variation distance (TVD) between the measurement probability distributions.

In addition to comparing with the original QAOA, we further compare the proposed EQE-QAOA with several existing qubit efficient schemes as benchmarks, including the minimal encoding scheme of Tan \textit{et al.}~\cite{tan2021qubit} (compact encoding), the partition-based method (DC-QAOA) of Li \textit{et al.}~\cite{li2022large} (problem partitioning), and the qubit-reuse compilation strategy of DeCross \textit{et al.}~\cite{decross2023qubit} (circuit-level compression). We consider a graph with $n$ vertices, a fixed-cardinality constraint $\sum_{i=1}^{n} x_i = k$ with $k = n/2$ is imposed across all methods to ensure a fair comparison of qubit efficiency and optimization performance. We evaluate (i) qubit requirement, defined as the number of qubits used; (ii) optimization performance, measured by the Max-Cut expectation gap $|\Delta E|$; and (iii) approximate ratio, defined as the ratio between the obtained solution value and the optimal value.

\subsection{Evaluation of EQE-QAOA}

\subsubsection{Qubit Reduction}

\begin{figure}[t]
  \centering
  \includegraphics[width=0.9\linewidth]{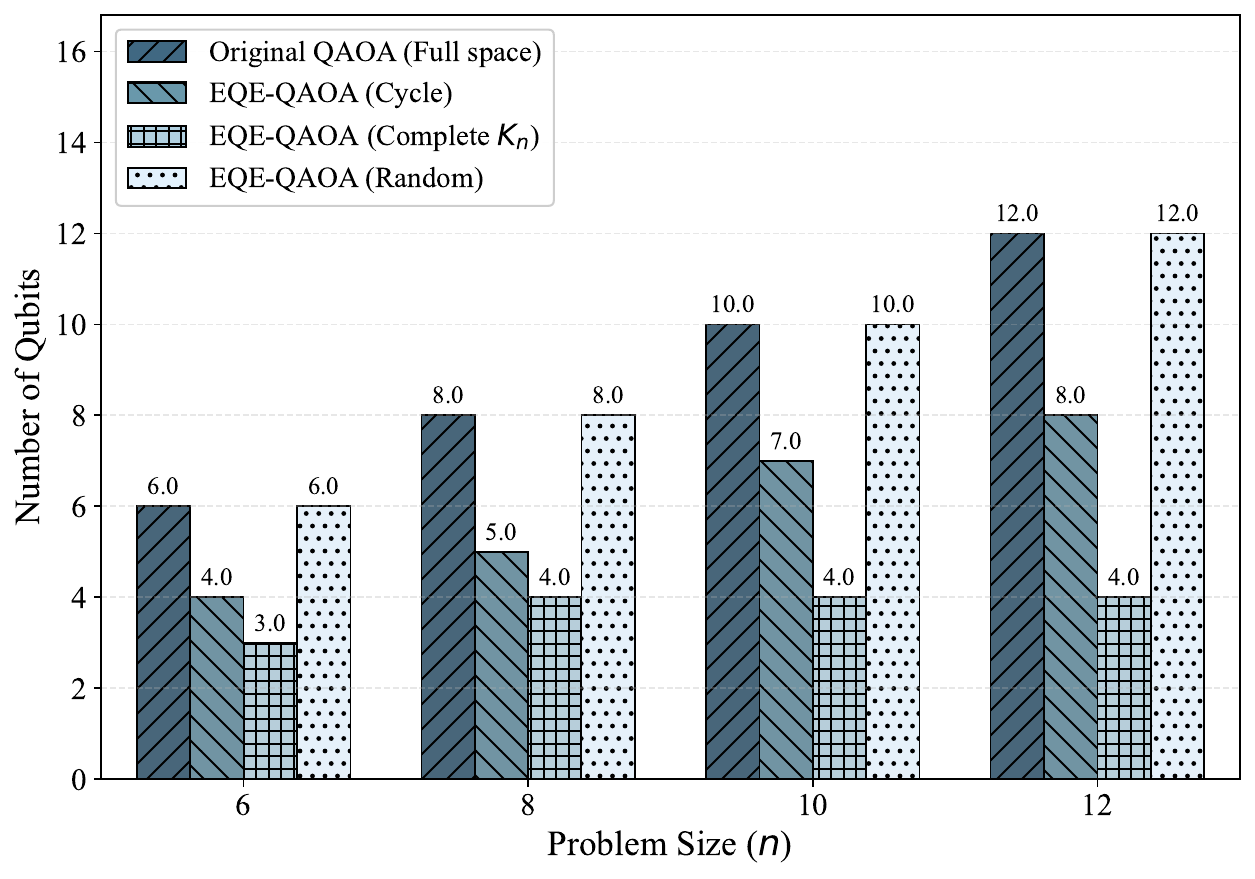}
  \caption{Qubit reduction achieved by EQE-QAOA across different graph families.}
  \label{fig_1}
\end{figure}

Fig.~\ref{fig_1} compares the original QAOA qubit requirement $n$ with the reduced requirement $m$ for cycle graphs, complete graphs $K_n$, and Erdős–Rényi random graphs as the problem size increases from $n=6$ to $n=12$. As shown in Fig.~\ref{fig_1}, the original QAOA scales linearly with $n$. By contrast, EQE-QAOA reduces the qubit requirement from $n$ to as few as $3$--$4$ qubits for graph families, resulting in up to a $70$\% reduction. In particular, complete graphs exhibit the strongest reduction, with the required qubit number remaining nearly constant across different problem sizes, reflecting their high degree of symmetry. Cycle graphs show moderate qubit reduction, indicating partial structural symmetry. In comparison, random graphs exhibit almost no reduction, with the required qubit number remaining close to $n$. This is indeed expected, as random graphs typically lack symmetry, leading to a Hamiltonian with an almost irreducible representation structure. Consequently, the effective invariant subspace dimension approaches that of the full Hilbert space, limiting the benefit of qubit reduction. Overall, these results show that problems exhibiting higher symmetry lead to greater qubit savings in EQE-QAOA, while less symmetric structures provide limited or no reduction.

\begin{figure}[t]
  \centering
  \includegraphics[width=0.9\linewidth]{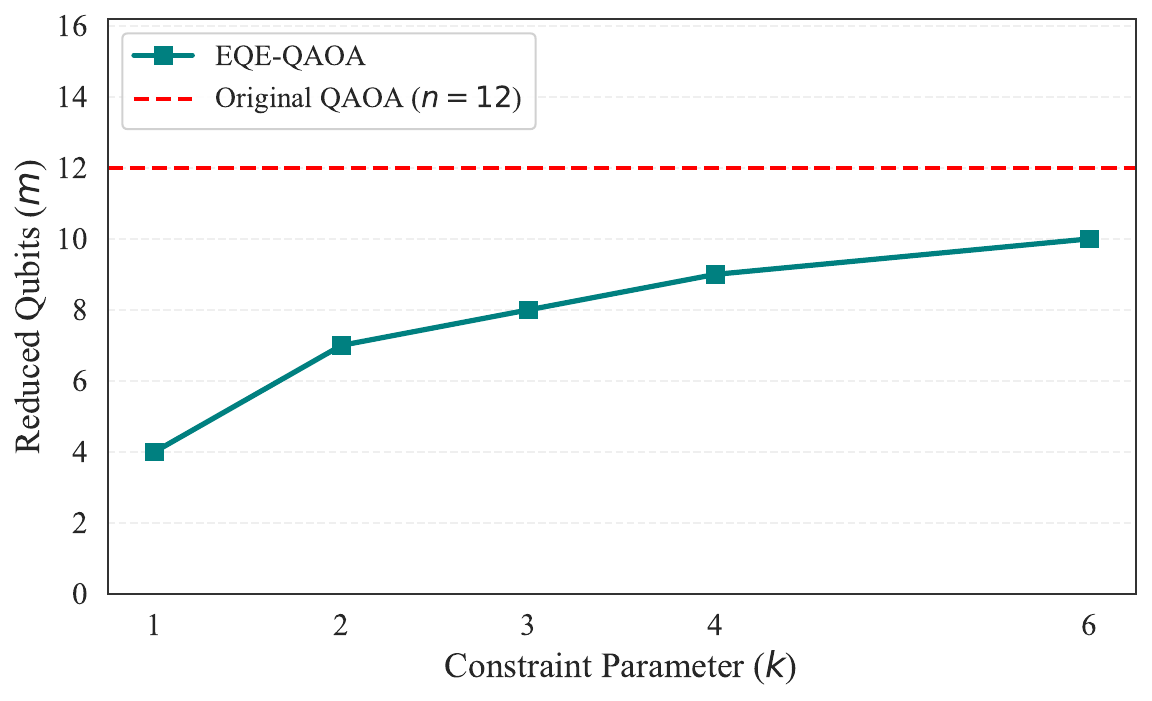}
  \caption{Effect of the Hamming-weight constraint $k$ on the reduced qubit requirement ($n=12$).}
  \label{fig_2}
\end{figure}

Fig.~\ref{fig_2} illustrates how the Hamming-weight constraint $\sum_{i=1}^{n} z_i = k$ affects the reduced qubit number $m$ for $n=12$. As shown in Fig.~\ref{fig_2}, 
$m$ increases with $k$ and tends towards the original requirement of 12 qubits. This is because smaller values of 
$k$ impose tighter constraints, which limit the number of feasible solutions and reduce the size of the feasible subspace. As a result, the effective dimension is smaller and fewer qubits are needed. As $k$ increases, the constraint becomes less restrictive, allowing more feasible solutions. This expands the feasible subspace, and therefore more qubits are required to represent it.

\subsubsection{Equivalence}

\begin{figure}[t]
\centering

\subfloat[State fidelity offset $F-1$.]{
    \includegraphics[width=0.9\linewidth]{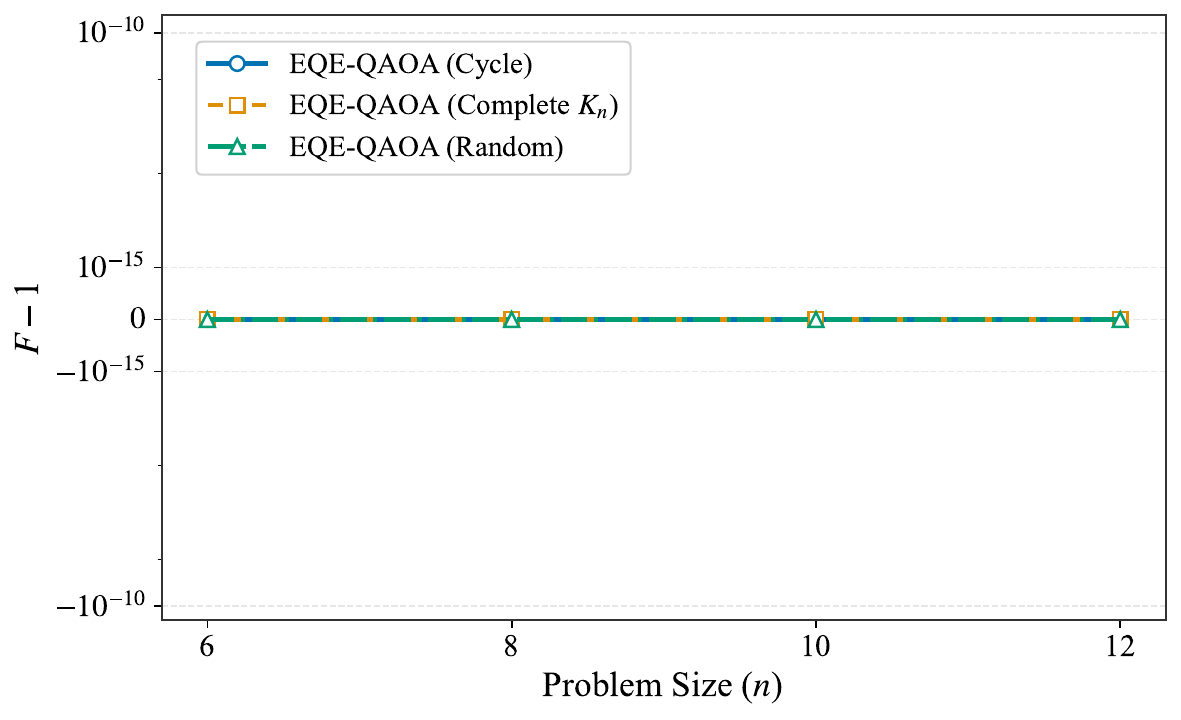}
}\\[-5pt]

\vspace{6pt}

\subfloat[Observable difference $|\Delta E|$.]{
    \includegraphics[width=0.9\linewidth]{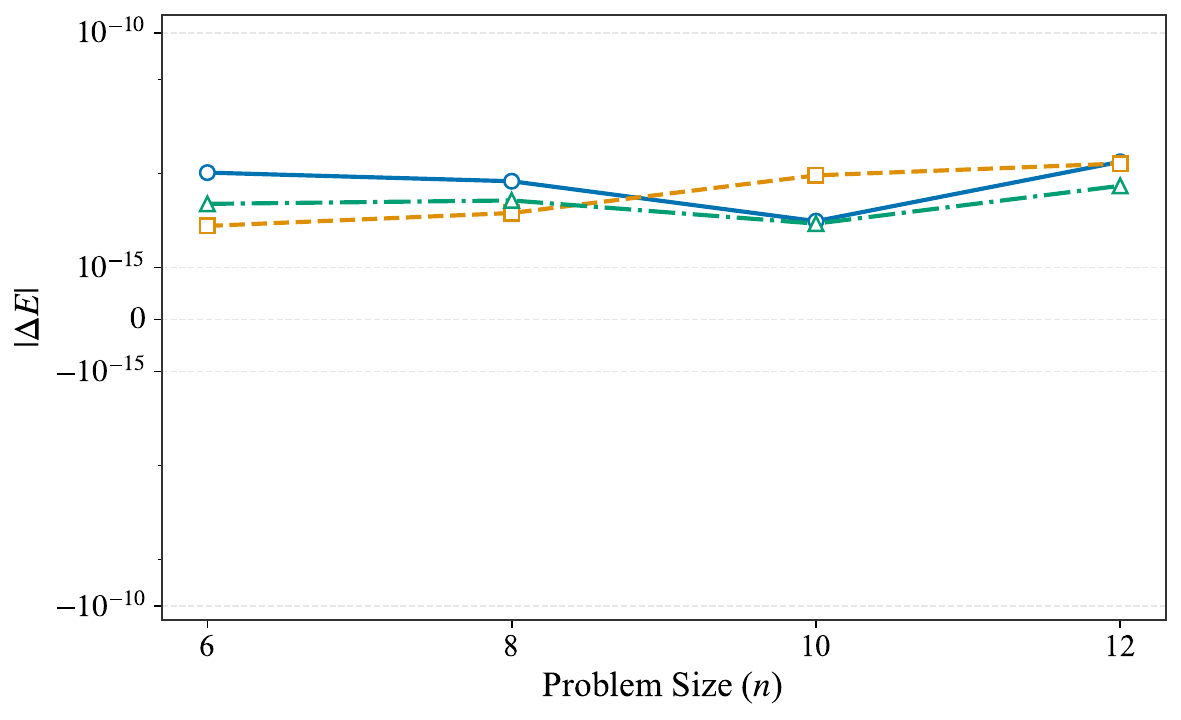}
}\\[-5pt]

\vspace{6pt}

\subfloat[Total variation distance (TVD).]{
    \includegraphics[width=0.9\linewidth]{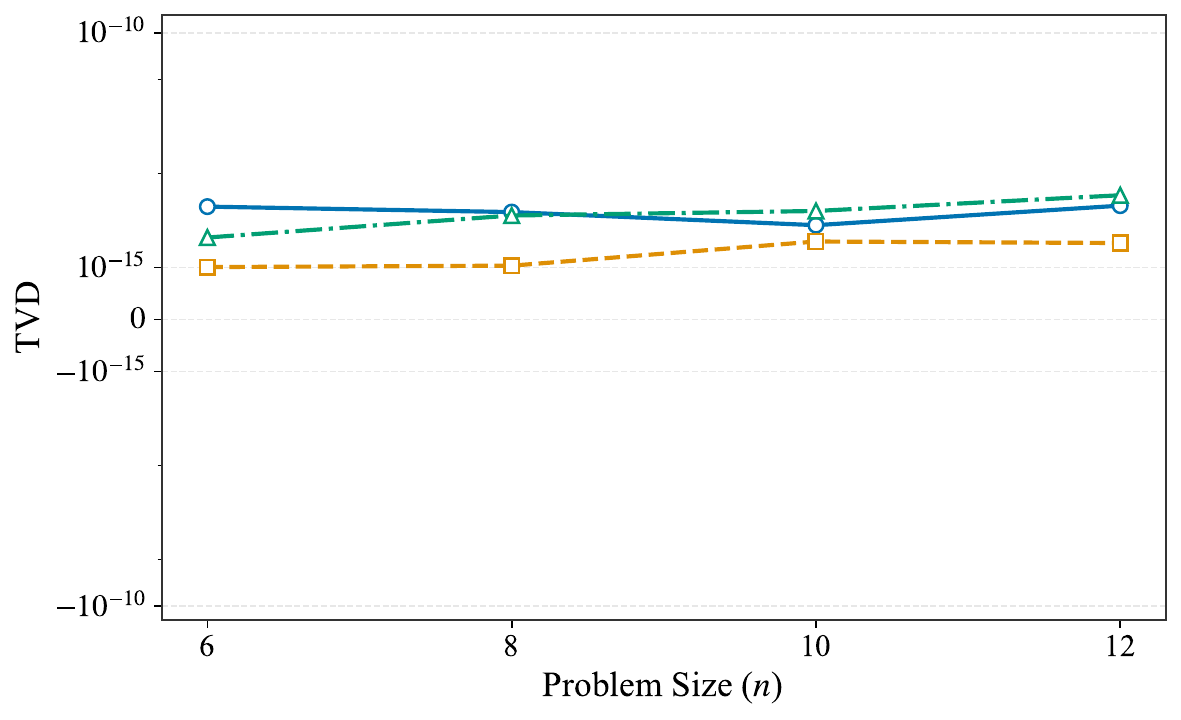}
}

\caption{Equivalence validation across general graph families.}
\label{fig:general_equivalence}
\end{figure}

Fig.~\ref{fig:general_equivalence} presents three equivalence metrics, namely state fidelity offset $(F-1)$, observable difference $|\Delta E|$, and TVD, for cycle, complete graphs $K_n$, and Erdős–Rényi random graphs as the problem size increases. The fidelity offset remains zero across all instances, indicating that EQE-QAOA and the original QAOA are indistinguishable within machine precision. Both the observable deviation and the probability TVD remain at the $10^{-14} \sim 10^{-15}$ level, with no systematic growth as $n$ increases. These results confirm that the reduced formulation preserves not only the quantum state evolution but also the expectation values and measurement statistics across different structural graph families.

\begin{figure}[t]
\centering

\subfloat[State fidelity offset $F-1$.]{
    \includegraphics[width=0.9\linewidth]{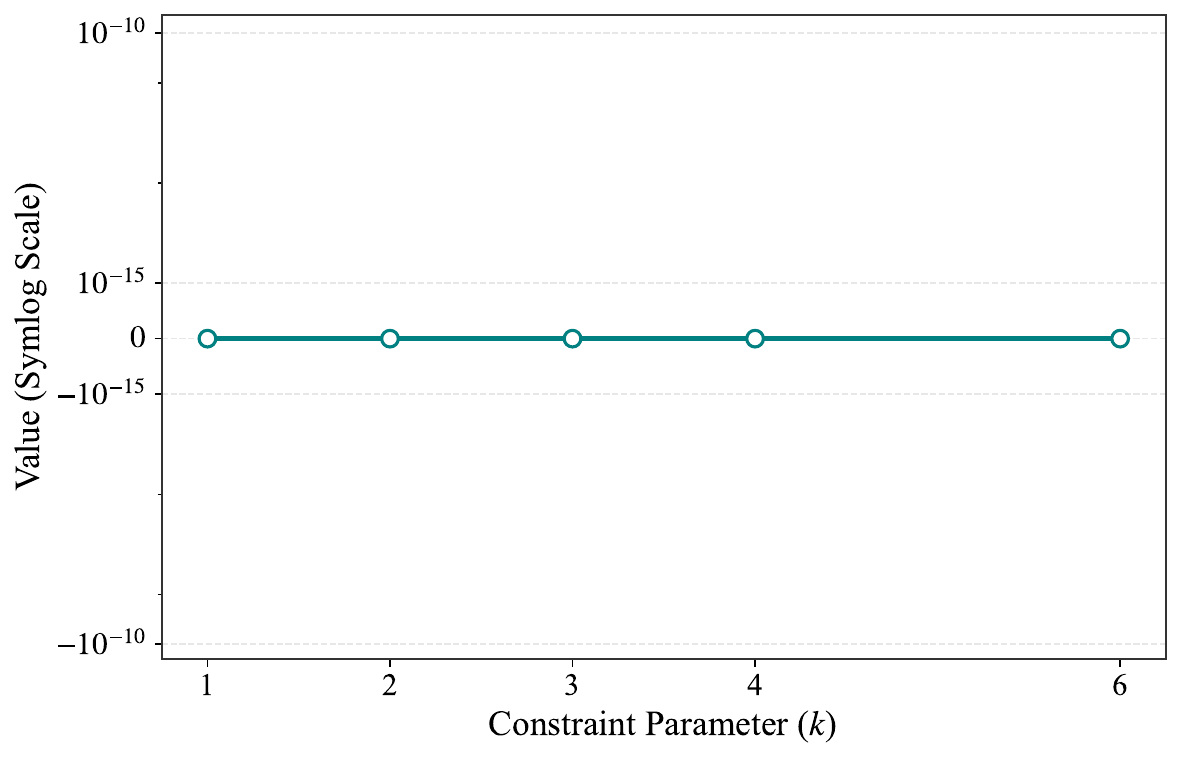}
}\\[-5pt]

\vspace{6pt}

\subfloat[Observable difference $|\Delta E|$.]{
    \includegraphics[width=0.9\linewidth]{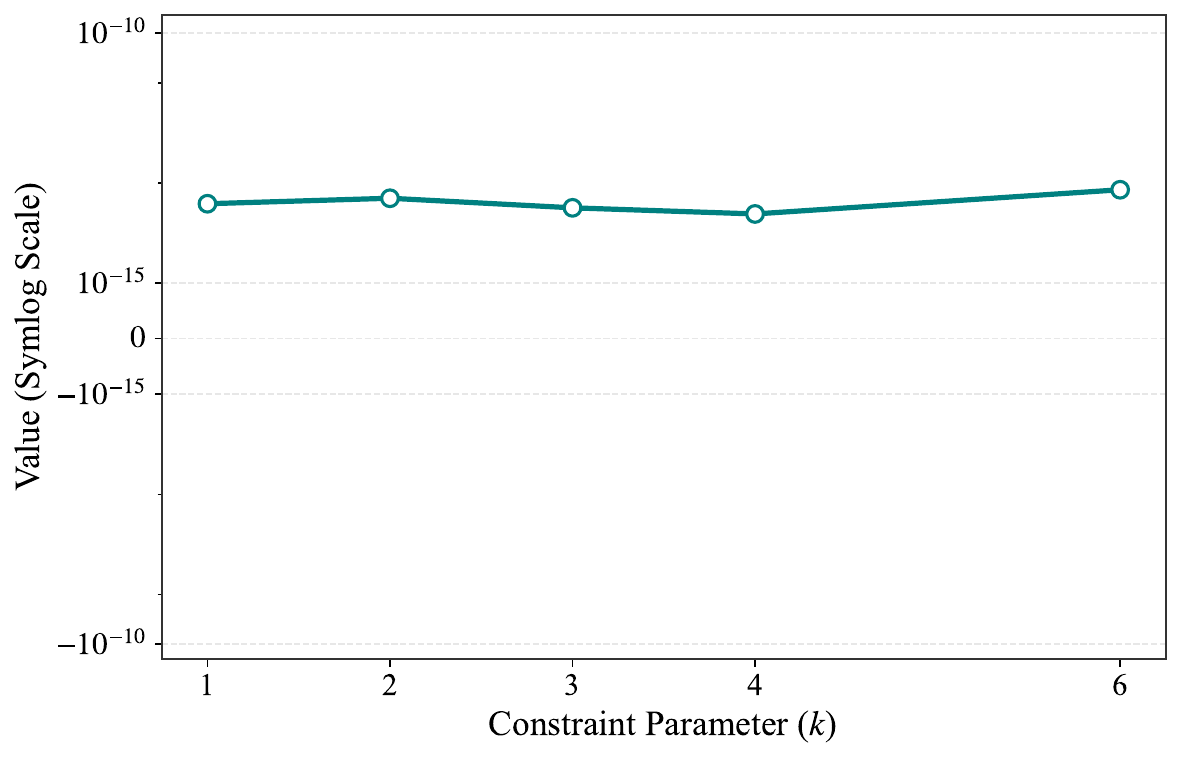}
}\\[-5pt]

\vspace{6pt}

\subfloat[Total variation distance (TVD).]{
    \includegraphics[width=0.9\linewidth]{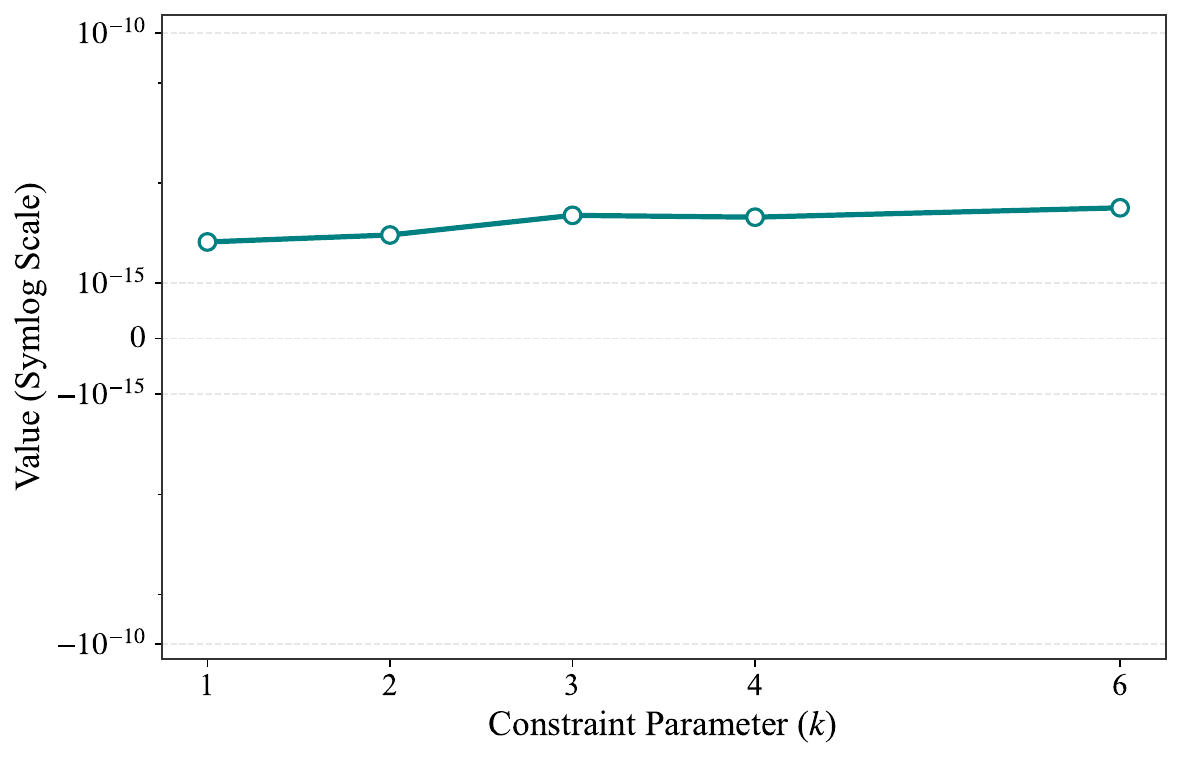}
}

\caption{Equivalence validation under the Random-$k$ constraint setting.}
\label{fig:randomk_equivalence}
\end{figure}

Fig.~\ref{fig:randomk_equivalence} presents the same three metrics for the constrained Max-Cut setting with fixed $n=12$ and varying Hamming weight $k$. Despite changes in the effective subspace dimension induced by different constraints, the fidelity offset remains zero and both $|\Delta E|$ and TVD stay at numerical precision across all tested $k$. This demonstrates that the reduction procedure remains exact under constrained feasible subspaces, preserving the dynamics and measurement behavior regardless of the Hamming-weight parameter.

\subsubsection{Memory}

\begin{figure}[t]
  \centering
  \includegraphics[width=0.9\linewidth]{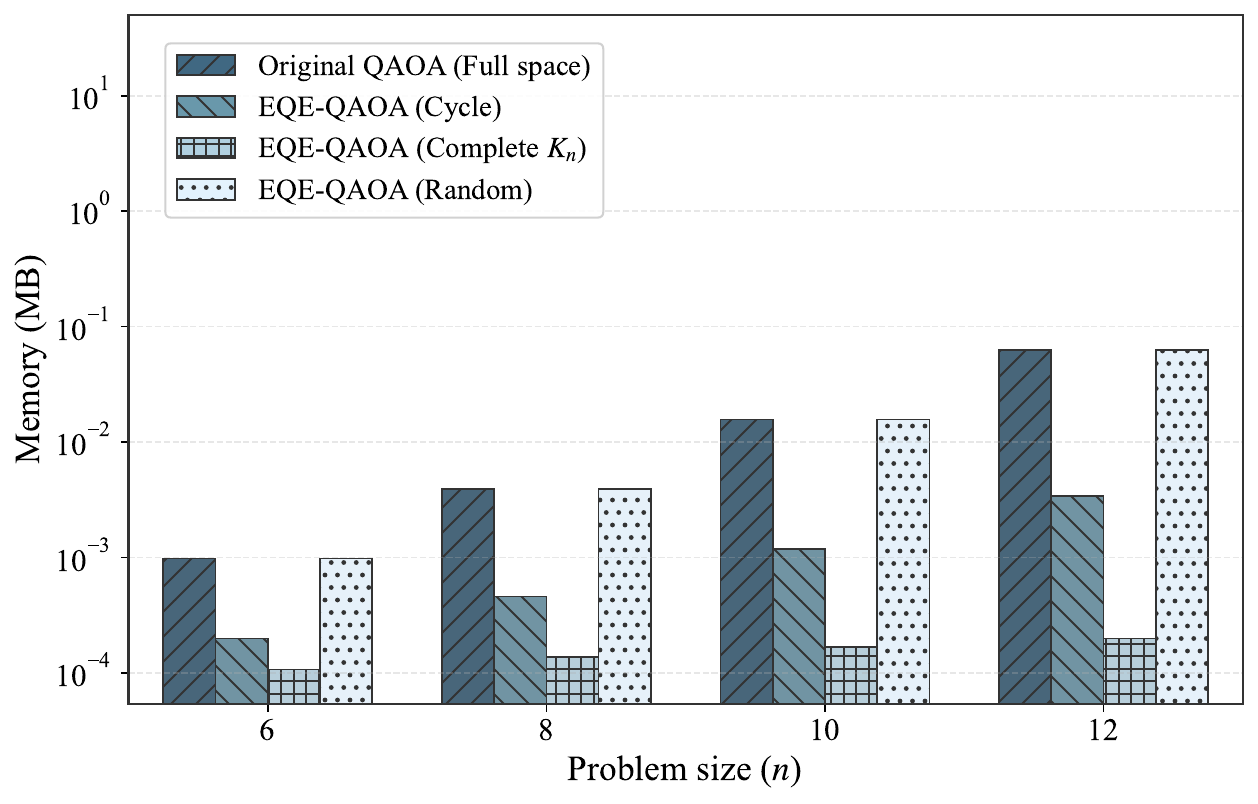}
  \caption{Comparison of state representation memory (Hilbert space dimension) between the original and reduced simulations.}
  \label{fig_5}
\end{figure}

Fig.~\ref{fig_5} compares the state representation memory of the original QAOA and EQE-QAOA on a logarithmic scale. As shown in Fig.~\ref{fig_5}, the memory requirement of the original QAOA increases exponentially with $n$, while EQE-QAOA grows much more slowly. The gap between the two becomes larger as $n$ increases, especially for graph families such as cycle and complete graphs. By contrast, for random graphs, the memory requirement of EQE-QAOA gradually approaches that of the original QAOA. This is because the original QAOA operates on the full Hilbert space of dimension $2^n$, leading to exponential memory growth. By contrast, EQE-QAOA operates on a reduced effective subspace with dimension $2^m$, which is much smaller for structured graphs due to their stronger symmetry. For random graphs, the lack of symmetry limits the reduction of the effective subspace, so the memory requirement becomes close to that of the original QAOA.

\subsection{Comparison with Existing Approaches}

\begin{figure}[t]
\centering

\subfloat[Qubit requirement vs $n$\label{fig:benchmark_qubit}]{
    \includegraphics[width=0.9\linewidth]{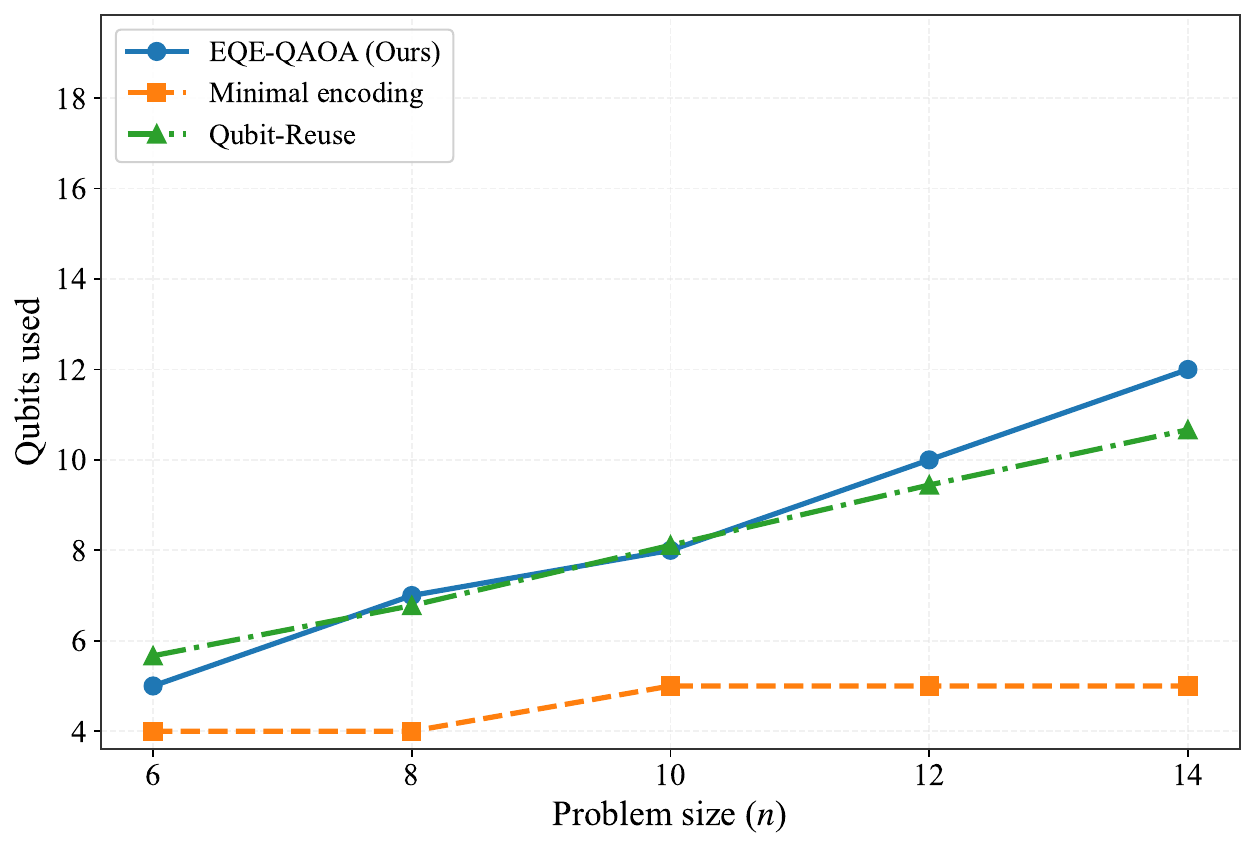}
}\\[-5pt]

\vspace{6pt}

\subfloat[Accuracy gap $|\Delta E|$ vs $n$\label{fig:benchmark_accuracy}]{
    \includegraphics[width=0.9\linewidth]{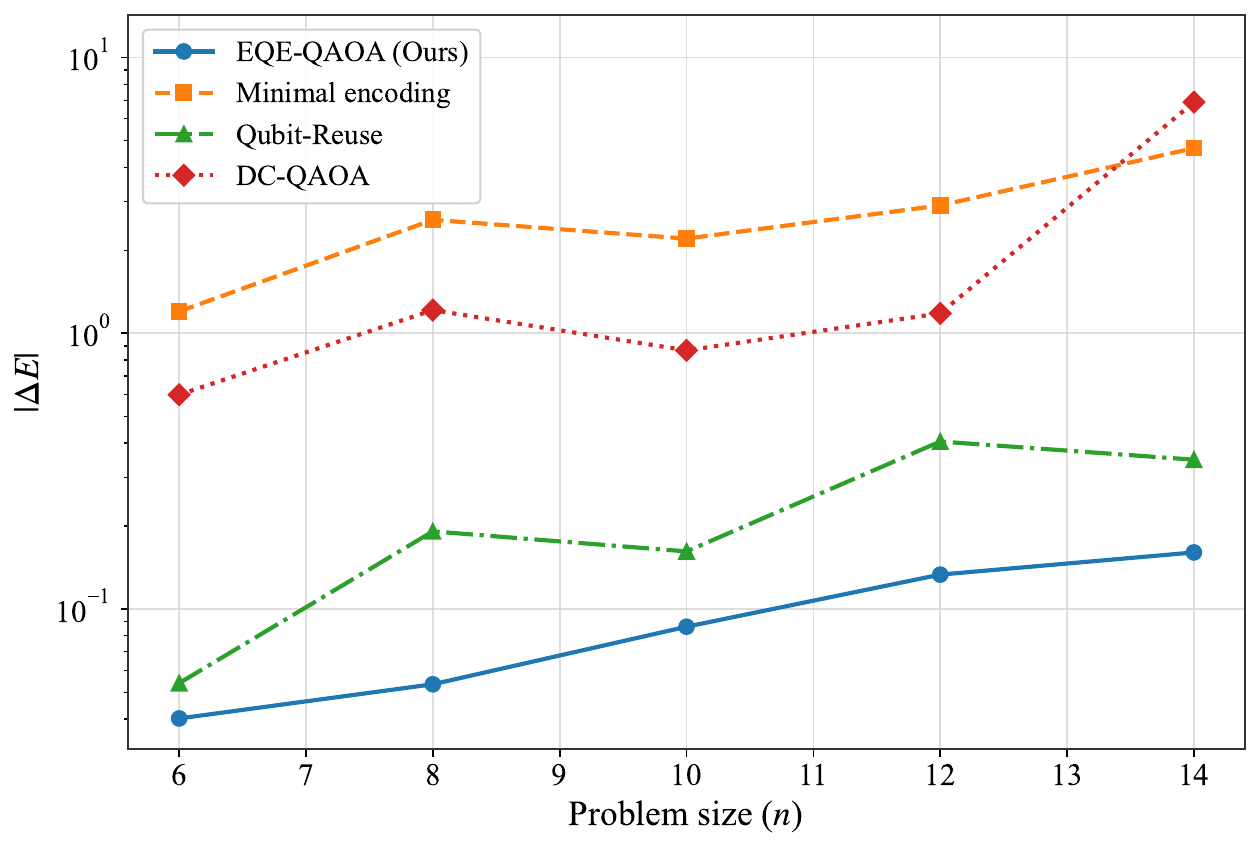}
}\\[-5pt]

\vspace{6pt}

\subfloat[Runtime vs $n$\label{fig:benchmark_runtime}]{
    \includegraphics[width=0.9\linewidth]{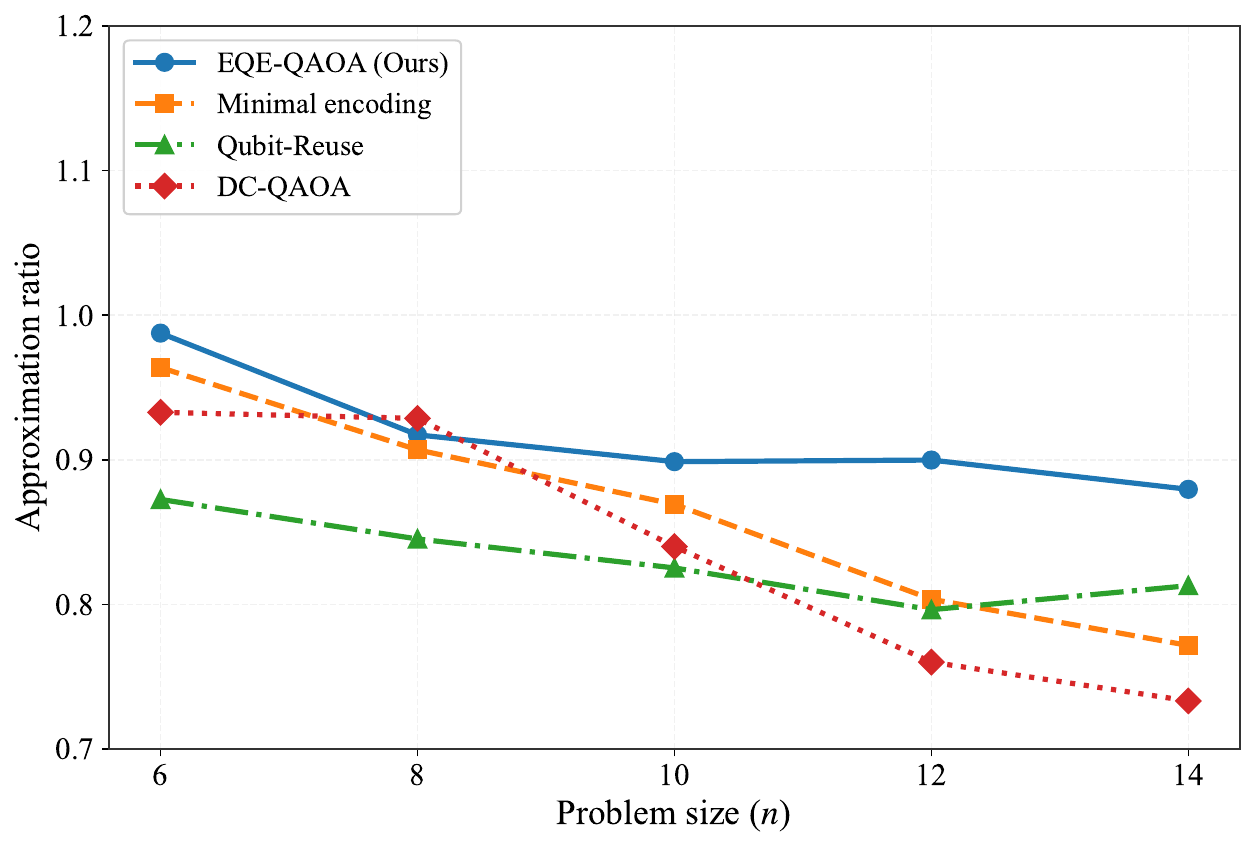}
}

\caption{Benchmark comparison of qubit-efficient methods.}
\label{fig:benchmark}
\end{figure}

To evaluate the proposed method, we compare EQE-QAOA with several representative qubit-efficient QAOA approaches. These methods reduce quantum resource requirements through different mechanisms: minimal encoding~\cite{tan2021qubit} compresses the representation of classical variables, DC-QAOA~\cite{li2022large} partitions the problem into smaller subproblems, and qubit-reuse compilation~\cite{decross2023qubit} reuses physical qubits during circuit execution. By contrast, EQE-QAOA exploits the invariant-subspace structure of the QAOA dynamics to construct a reduced space that preserves the original Hamiltonian evolution. The comparative results are summarized in Fig.~\ref{fig:benchmark}.

As shown in Fig.~\ref{fig:benchmark_qubit}, the minimal encoding scheme achieves the smallest qubit number under its most aggressive setting, where all variables are treated as independent. Despite this extreme compression, EQE-QAOA still attains a competitive level of qubit reduction compared to existing approaches. The qubit requirement of DC-QAOA is not included, because the number of qubits is typically predetermined by the available qubit resources, rather than being uniquely determined by the problem size. As shown in Fig.~\ref{fig:benchmark_accuracy}, EQE-QAOA consistently achieves the lowest accuracy gap across all problem sizes. The qubit-reuse method shows slightly larger gaps but remains close to EQE-QAOA, while the minimal encoding and DC-QAOA exhibit significantly larger gaps that increase with the problem size. Similarly, Fig.~\ref{fig:benchmark_runtime} shows that EQE-QAOA achieves the highest approximation ratio across all problem sizes. By contrast, other methods exhibit performance degradation as the problem size increases, and this gap becomes more significant for larger-scale instances.

These results show that EQE-QAOA provides the best overall performance by balancing qubit reduction with high accuracy. Compared with the other three methods, which sacrifice precision for qubit reduction, EQE-QAOA preserves the original QAOA evolution by exploiting the invariant subspace. This allows our method to achieve the smallest accuracy gap and the highest solution quality across all problem sizes.

\section{Conclusion}

The EQE-QAOA framework was proposed for qubit-efficient combinatorial optimization that preserves the original QAOA dynamics. It addressed the qubit bottleneck of QAOA by reducing the number of required qubits without affecting the optimization performance. Based on our analysis, the proposed EQE-QAOA is broadly applicable to large-scale constrained optimization problems in engineering. It reduced the required qubit resources, while also lowering computational complexity and improving scalability under limited qubit resources. Future work will focus on integrating qubit-efficient representations with the design of problem-specific mixer Hamiltonians, to further reduce computational cost and improve efficiency.


%





\ifCLASSOPTIONcaptionsoff
  \newpage
\fi


\appendices

\section{Proof of Lemma 1}

\textbf{Lemma 1.} Let $\Pi$ be an orthogonal projector acting on $\mathcal{H}$. If $[\Pi, H_{C,M}] = 0$, then $[\Pi, U(\boldsymbol{\gamma}, \boldsymbol{\beta})] = 0$ for any depth $p$.

\begin{proof}
    We first consider the commutation of $\Pi$ with the operator $e^{-i\gamma H_C}$. By definition of the matrix exponential, we have:%
    \begin{equation}
        e^{-i\gamma H_C} = \sum_{k=0}^{\infty} \frac{(-i\gamma)^k}{k!} H_C^k.
    \end{equation}

    \noindent Given $[\Pi, H_C] = 0$, it follows by induction that $\Pi H_C^k = H_C^k \Pi$ for all $k \in \mathbb{N}_0$. Specifically, for $k=0$, $\Pi I = I \Pi$ is trivial. Assuming the relation holds for $k-1$, we have:%
    \begin{equation}
        \Pi H_C^k = (\Pi H_C^{k-1}) H_C = (H_C^{k-1} \Pi) H_C = H_C^{k-1} (H_C \Pi) = H_C^k \Pi.
    \end{equation}

    \noindent By the linearity of the commutator, we obtain:%
    \begin{equation}
        [\Pi, e^{-i\gamma H_C}] = \sum_{k=0}^{\infty} \frac{(-i\gamma)^k}{k!} [\Pi, H_C^k] = 0.
    \end{equation}

    \noindent An identical argument yields $[\Pi, e^{-i\beta H_M}] = 0$. Next, consider a single QAOA layer $U_l(\gamma_l, \beta_l) = e^{-i\beta_l H_M} e^{-i\gamma_l H_C}$. Since $\Pi$ commutes with both exponential factors, we have:%
    \begin{equation}
        \begin{aligned}
        \Pi U_l 
        &= \Pi e^{-i\beta_l H_M} e^{-i\gamma_l H_C} \\
        &= e^{-i\beta_l H_M} \Pi e^{-i\gamma_l H_C} 
         = e^{-i\beta_l H_M} e^{-i\gamma_l H_C} \Pi 
         = U_l \Pi ,
        \end{aligned}
    \end{equation}

    \noindent hence $[\Pi, U_l] = 0$ for each layer $l$. Finally, the full QAOA unitary is the ordered product $U(\boldsymbol{\gamma}, \boldsymbol{\beta}) = \prod_{l=1}^p U_l(\gamma_l, \beta_l)$. Applying the commutation property for each layer iteratively, we get:%
    \begin{equation}
    \begin{aligned}
    \Pi U(\boldsymbol{\gamma}, \boldsymbol{\beta}) 
    &= \Pi (U_p U_{p-1} \cdots U_1) \\
    &= U_p \Pi (U_{p-1} \cdots U_1) = \cdots \\
    &= (U_p U_{p-1} \cdots U_1)\Pi 
     = U(\boldsymbol{\gamma}, \boldsymbol{\beta}) \Pi .
    \end{aligned}
    \end{equation}

    \noindent Therefore, $[\Pi, U(\boldsymbol{\gamma}, \boldsymbol{\beta})] = 0$.

\end{proof}

\section{Proof of Theorem 1}

\textbf{Theorem 1.} Assume the initial state $|\psi_0\rangle$ is contained within the effective subspace, i.e., $\Pi_{\mathrm{eff}}|\psi_0\rangle = |\psi_0\rangle$. Let $|\psi\rangle = U|\psi_0\rangle \in \mathcal{H}$ be the full-space evolved state and $|\psi_{\mathrm{eff}}\rangle = U_{\mathrm{eff}}|\psi_0\rangle \in \mathcal{H}_{\mathrm{eff}}$ be the reduced-space evolved state. Then:%
    \begin{enumerate}
        \item \textbf{State equivalence:} The state trajectory is confined to the effective subspace, $|\psi\rangle \in \mathcal{H}_{\mathrm{eff}}$, and $|\psi\rangle = |\psi_{\mathrm{eff}}\rangle$.

        \item \textbf{Observable equivalence:} For any observable $O$ that commutes with the subspace projector ($[O, \Pi_{\mathrm{eff}}] = 0$), the expectation values of $O$ satisfy:%
            \begin{equation}
                \langle\psi|O|\psi\rangle = \langle\psi_{\mathrm{eff}}| O|_{\mathcal{H}_{\mathrm{eff}}} |\psi_{\mathrm{eff}}\rangle.
            \end{equation}
            
        \item \textbf{Measurement equivalence:} For any solution $x \in \{0,1\}^n$, the probability distributions are identical:%
        \begin{equation}
            \Pr_{\mathrm{full}}(x) = \Pr_{\mathrm{eff}}(x).
        \end{equation}

        \noindent Specifically, if the basis state $|x\rangle$ is orthogonal to $\mathcal{H}_{\mathrm{eff}}$, then we have $\Pr_{\mathrm{full}}(x) = 0$.
    \end{enumerate}

    \begin{proof}
    \begin{enumerate}
        \item \textbf{State Equivalence:} The containment of the evolved state within $\mathcal{H}_{\mathrm{eff}}$ follows from the interplay between the initial condition and the symmetry of the unitary. Since $[\Pi_{\mathrm{eff}}, H_{C,M}] = 0$, by Lemma 1, the total unitary $U$ commutes with the projector: $[\Pi_{\mathrm{eff}}, U] = 0$. Given that $\Pi_{\mathrm{eff}}$ is an orthogonal projector, the initial state satisfies $|\psi_0\rangle = \Pi_{\mathrm{eff}}|\psi_0\rangle$. Thus, for the evolved state:%
        \begin{equation}
            |\psi\rangle = U\Pi_{\mathrm{eff}}|\psi_0\rangle = \Pi_{\mathrm{eff}}U|\psi_0\rangle.
        \end{equation}

        \noindent By the definition of a projector, this identity $\Pi_{\mathrm{eff}}|\psi\rangle = |\psi\rangle$ is a necessary and sufficient condition for $|\psi\rangle \in \mathcal{H}_{\mathrm{eff}}$.

        Furthermore, we examine the restriction equivalence. Each layer of the QAOA unitary can be expanded via the Taylor series. For any $|\phi\rangle \in \mathcal{H}_{\mathrm{eff}}$, we have:%
        \begin{equation}
            \begin{aligned}
            e^{-i\gamma H_C} |\phi\rangle
            &= \sum_{k=0}^{\infty} \frac{(-i\gamma)^k}{k!} H_C^k |\phi\rangle \\
            &= \sum_{k=0}^{\infty} \frac{(-i\gamma)^k}{k!} (H_C^{\mathrm{eff}})^k |\phi\rangle
            = e^{-i\gamma H_C^{\mathrm{eff}}} |\phi\rangle.
            \end{aligned}
        \end{equation}

        \noindent where the second equality holds because $H_C^k$ acts identically to $(H_C^{\mathrm{eff}})^k$ on the invariant subspace. By induction over the $P$ layers, the full-space trajectory is mapped directly to the reduced-space trajectory: $U|\psi_0\rangle = U_{\mathrm{eff}}|\psi_0\rangle$.

        \item \textbf{Observable Equivalence:} Consider an observable $O$ that commutes with the subspace projector, $[O, \Pi_{\mathrm{eff}}] = 0$. Since $|\psi\rangle = \Pi_{\mathrm{eff}}|\psi\rangle$, and the idempotency $\Pi_{\mathrm{eff}}^2 = \Pi_{\mathrm{eff}}$, we insert the identity $\Pi_{\mathrm{eff}}$ into the expectation value:%
        \begin{equation}
            \langle\psi| O |\psi\rangle = \langle\psi| \Pi_{\mathrm{eff}} O \Pi_{\mathrm{eff}} |\psi\rangle.
        \end{equation}

        \noindent Since $[O, \Pi_{\mathrm{eff}}] = 0$, we have $\Pi_{\mathrm{eff}} O \Pi_{\mathrm{eff}} = O \Pi_{\mathrm{eff}}^2 = O \Pi_{\mathrm{eff}}$. This operator is the formal extension of the restricted operator $O|_{\mathcal{H}_{\mathrm{eff}}}$ to the full space. Substituting the state identity $|\psi\rangle = |\psi_{\mathrm{eff}}\rangle$, we obtain:%
        \begin{equation}
            \langle\psi_{\mathrm{eff}}| (O \Pi_{\mathrm{eff}}) |\psi_{\mathrm{eff}}\rangle = \langle\psi_{\mathrm{eff}}| O|_{\mathcal{H}_{\mathrm{eff}}} |\psi_{\mathrm{eff}}\rangle,
        \end{equation}

        \noindent where the final step relies on the fact that $|\psi_{\mathrm{eff}}\rangle$ resides solely within the domain where $\Pi_{\mathrm{eff}}$ acts as identity.

        \item \textbf{Measurement Equivalence:} The probability of observing outcome $x$ in the full-space model is $\Pr_{\mathrm{full}}(x) = |\langle x | \psi\rangle|^2$. We decompose the basis state $|x\rangle$ as $|x\rangle = \Pi_{\mathrm{eff}}|x\rangle + (I - \Pi_{\mathrm{eff}})|x\rangle$. The transition amplitude becomes:%
        \begin{equation}
            \langle x | \psi\rangle = \langle x | \Pi_{\mathrm{eff}} |\psi\rangle = \left( \langle x | \Pi_{\mathrm{eff}} \right) |\psi\rangle.
        \end{equation}

        \noindent If $|x\rangle \perp \mathcal{H}_{\mathrm{eff}}$, then $\Pi_{\mathrm{eff}}|x\rangle = 0$, leading to zero probability. For states $|x\rangle$ with support in $\mathcal{H}_{\mathrm{eff}}$, the equivalence $|\psi\rangle = |\psi_{\mathrm{eff}}\rangle$ ensures the amplitudes are identical. Since the mapping from $\mathcal{H}_{\mathrm{eff}}$ to $\mathcal{H}$ is an isometric inclusion, the normalization is preserved, and the measurement statistics of the reduced system perfectly reconstruct those of the original problem.

    \end{enumerate}
\end{proof}

\section{Proof of Lemma 2}

\textbf{Lemma 2.} The operator $V$ satisfies the following identities:%
    \begin{equation}
       V^\dagger V = I_{\tilde{\mathcal{H}}_{\mathrm{eff}}}, \quad VV^\dagger = \Pi_{\mathrm{eff}}.
    \end{equation}

\begin{proof}
    By definition, $V = \sum_{k=0}^{M-1} |\phi_k\rangle\langle \tilde{\phi} _k|$, hence:%
    \begin{equation}
\begin{aligned}
V^\dagger V
&=
\left(\sum_{k=0}^{M-1} |\tilde{\phi} _k\rangle\langle \phi_k|\right)
\left(\sum_{j=0}^{M-1} |\phi_j\rangle\langle \tilde{\phi} _j|\right) \\
&=
\sum_{k,j=0}^{M-1} |\tilde{\phi} _k\rangle \langle \phi_k|\phi_j\rangle \langle \tilde{\phi} _j| \\
&=
\sum_{k=0}^{M-1} |\tilde{\phi} _k\rangle\langle \tilde{\phi} _k| \\
&=
I_{\tilde{\mathcal H}_{\mathrm{eff}}},
\end{aligned}
\end{equation}

\noindent where we used the orthonormality $\langle \phi_k|\phi_j\rangle=\delta_{kj}$. Similarly, %
\begin{equation}
\begin{aligned}
VV^\dagger
&=
\left(\sum_{k=0}^{M-1} |\phi_k\rangle\langle \tilde{\phi} _k|\right)
\left(\sum_{j=0}^{M-1} |\tilde{\phi} _j\rangle\langle \phi_j|\right) \\
&=
\sum_{k,j=0}^{M-1} |\phi_k\rangle \langle \tilde{\phi} _k|\tilde{\phi} _j\rangle \langle \phi_j| \\
&=
\sum_{k=0}^{M-1} |\phi_k\rangle\langle \phi_k| \\
&=
\Pi_{\mathrm{eff}},
\end{aligned}
\end{equation}

\noindent where we used $\langle \tilde{\phi} _k|\tilde{\phi} _j\rangle=\delta_{kj}$ and the fact that $\{|\phi_k\rangle\}_{k=0}^{M-1}$ is an orthonormal basis of $\mathcal H_{\mathrm{eff}}$.

\end{proof}

\section{Proof of Theorem 2}

\textbf{Theorem 2.} For any parameters $\gamma, \beta \in \mathbb{R}$ and the isometric mapping $V$, the following operator identities hold:%
    \begin{equation}
        e^{-i\gamma H_C} V = V e^{-i\gamma \tilde{H}_{C}^{\mathrm{eff}}}, \quad e^{-i\beta H_M} V = V e^{-i\beta \tilde{H}_{M}^{\mathrm{eff}}}.
    \end{equation}

    \noindent Consequently, for a QAOA circuit with $P$ layers, the total evolution operators $U(\boldsymbol{\gamma}, \boldsymbol{\beta})$ and $\tilde{U}(\boldsymbol{\gamma}, \boldsymbol{\beta})$ satisfy:%
    \begin{equation}
        U(\boldsymbol{\gamma}, \boldsymbol{\beta}) V = V \tilde{U}(\boldsymbol{\gamma}, \boldsymbol{\beta}).
    \end{equation}

\begin{proof}
    Consider the cost Hamiltonian $H_C$. Using the identity $V V^\dagger = \Pi_{\mathrm{eff}}$ from Lemma 2 and the fact that $H_C$ stays within the effective subspace, we see that:%
    \begin{equation}
        H_C V = (V V^\dagger) H_C V = V (V^\dagger H_C V) = V \tilde{H}C.
    \end{equation}

    \noindent This identity shows that the action of $H_C$ on the mapped states is exactly captured by $\tilde{H}_C$ in the small space. This also works for any power $k \ge 0$, so $H_C^k V = V \tilde{H}_C^k$. By using the series expansion for the exponential function, we get:%
    \begin{equation}
        e^{-i\gamma H_C} V = \sum_
{k=0}^\infty \frac{(-i\gamma)^k}{k!} H_C^k V = \sum_{k=0}^\infty \frac{(-i\gamma)^k}{k!} V \tilde{H}_C^k = V e^{-i\gamma \tilde{H}_C}.
    \end{equation}

    \noindent The same logic applies to $H_M$. Since this works for every single layer, it works for the entire QAOA circuit. 

\end{proof}

\section{Proof of Corollary 1}

\textbf{Corollary 1.} Under the assumptions of Theorem~1, for any reduced initial state $|\tilde\psi_0\rangle$ with $|\psi_0\rangle = V|\tilde\psi_0\rangle$, the dynamics generated in the reduced $m$-qubit space are equivalent to the original QAOA dynamics in $\mathcal H$. Consequently, the reduced system reproduces the same state trajectory, expectation values, and measurement statistics as the full system.

\begin{proof}
    \begin{enumerate}
        \item \textbf{State Matching:} By Theorem 1, the full-space trajectory is restricted to the effective subspace, $|\psi\rangle \in \mathcal{H}_{\text{eff}}$. Since $V V^\dagger = \Pi_{\text{eff}}$ acts as the identity on $\mathcal{H}_{\text{eff}}$, we have $|\psi\rangle = V V^\dagger |\psi\rangle$. From Theorem 2, the correspondence $V^\dagger U = \tilde{U} V^\dagger$ holds. Applying this to the initial state yields:%
        \begin{equation}
            |\psi\rangle = V \tilde{U} V^\dagger |\psi_0\rangle = V \tilde{U} |\tilde\psi_0\rangle = V |\tilde\psi\rangle.
        \end{equation}

        \item \textbf{Observable Matching:} From Theorem 1, the expectation value of any observable $O$ commuting with $\Pi_{\text{eff}}$ depends only on its restriction to the subspace. By substituting the state matching result $|\psi\rangle = V |\tilde\psi\rangle$ into the expectation value expression:%
        \begin{equation}
            \langle\psi|O|\psi\rangle = \langle\tilde\psi| V^\dagger O V |\tilde\psi\rangle.
        \end{equation}

        \noindent This confirms that the reduced operator $\tilde{O} = V^\dagger O V$ preserves the physical statistics.

        \item \textbf{Measurement Matching:} Since Theorem 1 ensures $\Pr_{\text{full}}(x) = \Pr_{\text{eff}}(x)$, and the isometry $V$ defines a unitary mapping between the orthonormal bases $\{|\phi_k\rangle\}$ of $\mathcal{H}_{\text{eff}}$ and $\{|\tilde{\phi} _k\rangle\}$ of $\tilde{\mathcal{H}}_{\text{eff}}$, the probability distributions over these corresponding indices must be identical. Formally, for $x$ such that $|x\rangle \in \mathcal{H}_{\text{eff}}$, we have:%
        \begin{equation}
            \langle x | \psi \rangle = \langle x | V |\tilde\psi\rangle = \sum_k \langle x | \phi_k \rangle \langle \tilde{\phi} _k | \tilde\psi \rangle.
        \end{equation}

        \noindent In the computational basis case where $|x\rangle = |\phi_k\rangle$, this reduces to $\langle x | \psi \rangle = \langle \tilde{\phi} _k | \tilde\psi \rangle$, yielding identical probabilities.
    \end{enumerate}
\end{proof}


%

\bibliographystyle{IEEEtran}
\bibliography{reference}

%




\end{document}